\documentclass[11pt]{article}

\usepackage{graphicx}
\usepackage{dcolumn}
\usepackage{bm}

\usepackage{mathrsfs}  

\usepackage{appendix}

\usepackage[utf8]{inputenc}
\usepackage[T1]{fontenc}
\usepackage{color}
\usepackage{physics}
\usepackage{fullpage}
\usepackage{amsthm}
\usepackage{amssymb}
\usepackage{mathtools}
\usepackage{authblk}
\usepackage{bbm}
\usepackage[colorlinks,linkcolor=blue,citecolor=blue,urlcolor=magenta,filecolor=cyan]{hyperref}
\usepackage[capitalize]{cleveref}

\newtheorem{theorem}{Theorem}
\newtheorem{lemma}[theorem]{Lemma}

\newtheorem{proposition}[theorem]{Proposition}
\newtheorem{corollary}[theorem]{Corollary}

\newtheorem{remark}[theorem]{Remark}

\def\eps{\varepsilon}

\def\polylog{\mathrm{polylog}}
\def\poly{\mathrm{poly}}

\def\bbc{\mathbb{C}}
\def\bbr{\mathbb{R}}

\newcommand{\comment}[1]{}

\allowdisplaybreaks[2]
\newcommand{\iu}{\mathrm{i}}
\newcommand{\id}{I}
\newcommand{\HH}{\mathcal H}
\newcommand{\EE}{\mathcal E}
\newcommand{\QQ}{\mathcal Q}
\newcommand{\PP}{\mathcal P}
\newcommand{\diag}{\operatorname{diag}}
\newcommand{\Span}{\operatorname{span}}

\title{Sparse Hamiltonian simulation with optimal dependence on the maximum column Euclidean norm}

\author{Zecheng Li}
\author{Chunhao Wang}

\affil{Department of Computer Science and Engineering, Pennsylvania State University}
\affil[]{Email: \{zxl5523,cwang\}@psu.edu}

\date{}

\begin{document}
\maketitle

\begin{abstract}
We give a quantum algorithm for simulating a $d$-sparse Hermitian Hamiltonian $H$, assuming known upper bounds $\Lambda$ on its maximum column Euclidean norm $\norm{H}_{1\to2}$ and $\Lambda_{\max}$ on its maximum entry magnitude.  For every $t\ge0$, simulation with operator-norm error $\eps$ uses
\[
 O\!\left(t\Lambda\sqrt d
 +\sqrt{td\Lambda_{\max}\log(2/\epsilon)}+\log(2/\epsilon)\right)
\]
sparse-oracle queries. This removes the subpolynomial overhead in Low's algorithm [STOC 2019]. For $d>1$ and $t\Lambda\ge\log(2/\epsilon)$, the bound matches the worst-case lower bound.  A known spectral-norm upper bound may also be used in place of $\Lambda$.  The number of 1- and 2-qubit gates is linear in the query scale, up to oracle costs and polynomial overhead in the input bit lengths and logarithmic precision parameters.

As applications, we obtain $O(\kappa\sqrt d\,\mathrm{polylog}(\kappa/\epsilon))$ queries for solving $d$-sparse quantum linear systems with $\norm{A}\le1$ and $\norm{A^{-1}}\le\kappa$, under standard sparse and state-preparation access.  We also give a gate-efficient implementation of black-box unitaries with at most $d$ nonzero entries per row and column using $O(\sqrt{d\log(2/\epsilon)}+\log(2/\epsilon))$ queries, given sparse access to the unitary and its adjoint.  At constant error, the query bound is optimal and yields $\Theta(\sqrt N)$ queries for arbitrary $N\times N$ unitaries, resolving the open question on black-box unitary implementation posed by Berry and Childs [QIC 2012].
\end{abstract}

\section{Introduction}
\subsection{Background and motivation}

Hamiltonian simulation was one of the original motivations for quantum
computing \cite{Feynman1982,Lloyd1996}.  A time-independent quantum
system is described by a Hermitian matrix $H$.  At time
$t$, its state obeys the Schr\"odinger equation
\begin{equation}
 \iu\frac{d}{dt}\ket{\psi(t)}=H\ket{\psi(t)}.
\end{equation}
The solution after evolution time $t$ is
$\ket{\psi(t)}=e^{-\iu tH}\ket{\psi(0)}$.  The Hamiltonian simulation
problem is to implement unitary $e^{-\iu tH}$, to a prescribed error, using a
quantum circuit and an efficient description of $H$.  In addition to
modeling dynamics directly in chemistry, materials, and many-body physics
\cite{GeorgescuAshhabNori2014,BauerEtAl2020}, Hamiltonian simulation is
a central candidate application for fault-tolerant quantum computers.

The algorithmic study of Hamiltonian simulation has progressed through
product formulas
\cite{Lloyd1996,AharonovTaShma2003,BerryAhokasCleveSanders2007},
sparse quantum walks \cite{Childs2010,BerryChilds2012}, linear
combinations of unitaries \cite{ChildsWiebe2012,BerryChildsKothari2015},
and truncated Taylor series \cite{BerryEtAl2015}.  Quantum signal
processing and qubitization subsequently gave essentially optimal
dependence on the simulation time and precision, while quantum singular
value transformation placed these methods in a broader matrix-arithmetic
framework \cite{LowChuang2017,LowChuang2019,Gilyen2019}.  Hamiltonian
simulation is also a computational primitive rather than only a tool for
physical simulation.  It underlies algorithms for eigenvalue
estimation, quantum linear systems, and the implementation of a unitary
specified by its matrix entries
\cite{Kitaev1995,HarrowHassidimLloyd2009,ChildsKothariSomma2017,
BerryChilds2012}.

In this paper, we assume access to the standard sparse-access model.  Let
$\norm{H}_{\max}:=\max_{i,j}|H_{ij}|$.  A direct sparse block encoding
has normalization on the scale $d\norm{H}_{\max}$
\cite[Lemma~48]{Gilyen2019}, and qubitization
therefore gives a simulation cost essentially linear in
$td\norm{H}_{\max}$, apart from its precision dependence
\cite{LowChuang2019}.  This is the
right description when the maximum entry is the available physical
scale.  It can be unfavorable, however, in applications in which $H$
is normalized by its spectral norm.  Merely substituting
$\norm{H}_{\max}\le\norm{H}$ gives a dependence on $td\norm{H}$, even
though the black-box lower bound permits a square-root improvement in
$d$ \cite{BerryChilds2012,Low2019}.

The norm that resolves this mismatch is the \emph{maximum column Euclidean norm}:
\begin{equation}
 \norm{H}_{1\to2}:=\max_j\norm{H\ket{j}}_2
 =\max_j\left(\sum_i|H_{ij}|^2\right)^{1/2}.
\label{eq:intro-one-two}
\end{equation}
For simplicity, we also call this the $1$-to-$2$ norm.
For Hermitian $H$, it is also the
maximum row Euclidean norm.  If $H$ is $d$-sparse, then
\begin{equation}
 \norm{H}_{\max}
 \le \norm{H}_{1\to2}
 \le \norm{H}
 \le \sqrt d\,\norm{H}_{1\to2}
 \le d\norm{H}_{\max}.
 \label{eq:intro-norm-chain}
\end{equation}
Thus the scale $\sqrt d\,\norm{H}_{1\to2}$ interpolates between the
max-entry and spectral-norm descriptions and is never worse than the
usual sparse normalization $d\norm{H}_{\max}$.

Low~\cite[Theorem~2]{Low2019} first obtained a sparse-simulation algorithm with this stronger
norm dependence.  Given
$\Lambda\ge\norm{H}_{1\to2}$ and
$\tau=t\Lambda\sqrt d$, its query complexity is
\begin{equation}
 O\!\left(
  \tau\bigl(\log(\tau/\eps)\bigr)^{O(\sqrt{\log d})}
 \right)
 =O\!\left(\tau(\tau/\eps)^{o(1)}\right).
 \label{eq:intro-low-bound}
\end{equation}
The same work exhibited instances requiring
$\Omega(t\sqrt d\,\norm{H}_{1\to2})$ queries at bounded error
\cite[Theorem~3]{Low2019}.  Hence,
when the promised bound $\Lambda$ is tight up to constants, the time,
sparsity, and norm dependence was known up to subpolynomial factors,
but it remained open whether those factors could be removed while
preserving the $1$-to-$2$ norm dependence.  We remove this overhead
for every evolution time $t\ge0$ and give an efficient circuit
realization with explicit gate and qubit costs.  Optimizing the
working normalization gives the mixed-term bound in \cref{eq:query-bound}.  When
$t\Lambda\ge\log(2/\eps)$, it simplifies to $O(t\sqrt d\,\Lambda)$
queries, matching the worst-case lower bound for $d>1$.  This
condition concerns optimality, not the validity of the algorithm.

\subsection{Access model and main results}
\label{sec:model}

Write $[N]:=\{1,\ldots,N\}$ and set $\HH:=\bbc^N$.  Let
$H\in\bbc^{N\times N}$ be Hermitian, with at most $d$ nonzero entries
in each row.  We assume $1\le d\le N$, an evolution time $t\ge0$, and
a known number $\Lambda>0$ satisfying
\begin{equation}\label{eq:one-two-promise}
 \norm{H}_{1\to2}\le\Lambda.
\end{equation}
We also allow a known upper bound $\Lambda_{\max}>0$ on the maximum
entry magnitude:
\begin{equation}\label{eq:max-entry-promise}
 \norm{H}_{\max}\le\Lambda_{\max}\le\Lambda.
\end{equation}
The last inequality entails no loss of generality: any supplied 
bound may be replaced by its minimum with $\Lambda$.  If no separate
bound is available, take $\Lambda_{\max}=\Lambda$.
We use the normalized time $t\Lambda$ and the sparsity-amplified
simulation time
\begin{equation}\label{eq:parameters}
 \tau=t\sqrt d\,\Lambda.
\end{equation}

For each row $i$, let $f_i$ be a permutation of $[N]$ whose first $d$
values include every nonzero position in row $i$.  If that row has
fewer than $d$ nonzero entries, the remaining positions are distinct
zero entries.  Following the standard sparse-access model
\cite{BerryAhokasCleveSanders2007,Low2019}, the input is supplied by
the sparse-access oracles
\begin{align}
 O_F\ket{i,\ell}&=\ket{i,f_i(\ell)},\label{eq:OF}\\
 O_H\ket{i,j,y}&=\ket{i,j,y\oplus H_{ij}}.
 \label{eq:OH}
\end{align}
Each matrix entry is supplied to $b_{\rm val}$ bits of precision in
the value register.  Calls to these oracles, their inverses, and
constant-cost controlled versions are counted.  

We require simulation with operator-norm error $\eps$.  If $\mathcal A_H$
is the implemented unitary and $\ket{0}$ denotes all clean ancillary
registers, then
\begin{equation}\label{eq:simulation-error}
 \norm{\mathcal A_H(\ket{0}\otimes\id)
      -\ket{0}\otimes e^{-\iu tH}}
 \le \eps.
\end{equation}
This operator norm between input--output isometries also controls an
input entangled with an arbitrary reference system.

Our query bound is as follows.

\begin{theorem}[Sparse Hamiltonian simulation from norm bounds]
\label{thm:main}
For every input satisfying the preceding assumptions and every
$t\ge0$ and $0<\eps\le1/2$, there is a circuit
satisfying \cref{eq:simulation-error} with query complexity
\begin{equation}
 O\!\left(
 t\sqrt d\,\Lambda
 +\sqrt{td\Lambda_{\max}\log\frac{2}{\eps}}
 +\log\frac{2}{\eps}
 \right).
 \label{eq:query-bound}
\end{equation}
In particular, if $t\Lambda\ge\log(2/\eps)$, this bound simplifies
to $O(\tau)$.
\end{theorem}

Compared with \cref{eq:intro-low-bound}, the theorem removes the
subpolynomial overhead and makes the dependence on precision and the
maximum entry magnitude explicit.  Since $\Lambda_{\max}\le\Lambda$,
it also implies a simpler bound
\begin{equation}\label{eq:original-query-bound}
 O\!\left(\tau+\sqrt d\log\frac{2}{\eps}\right).
\end{equation}
The theorem applies to all $t\ge0$.  When
$t\Lambda\ge\log(2/\eps)$, the leading term dominates and the entire
query complexity is $O(t\Lambda\sqrt d)$.  For $d>1$, Low's lower bound applies already at
constant error and includes instances with
$\Lambda=\Theta(\norm{H}_{1\to2})$ \cite[Theorem~3]{Low2019}.
Thus the leading $t\Lambda\sqrt d$ dependence matches the
dimension-independent worst-case lower bound, allowing the dimension
$N$ to grow with the other parameters, and the upper bound is
worst-case optimal in the stated regime.  We do not claim joint
optimality of all terms in \cref{eq:query-bound} outside this regime.

The result in \cref{thm:main} is stronger than a spectral-norm
formulation: the $1$-to-$2$ norm promise is weaker, since every column
Euclidean norm is at most the spectral norm.  This gives the following
immediate consequence.

\begin{corollary}[Spectral-norm formulation]
\label{cor:spectral-norm}
Suppose instead that a number $\Gamma>0$ with $\Gamma\ge\norm{H}$ is
known, together with an optional bound
$\norm{H}_{\max}\le\Lambda_{\max}\le\Gamma$; without a separate
entry bound, take $\Lambda_{\max}=\Gamma$.  For every $t\ge0$,
simulation with operator-norm error
$0<\eps\le1/2$ uses
\begin{equation}\label{eq:spectral-corollary}
 O\!\left(
 t\sqrt d\,\Gamma
 +\sqrt{td\Lambda_{\max}\log\frac{2}{\eps}}
 +\log\frac{2}{\eps}
 \right)
\end{equation}
queries.  If $t\Gamma\ge\log(2/\eps)$, the query complexity is
$O(t\sqrt d\,\Gamma)$.
\end{corollary}

\begin{proof}
For Hermitian $H$, every column Euclidean norm is at most
$\norm{H}$.  Thus $\Gamma$ is a valid $1$-to-$2$ norm bound in
\cref{thm:main}; apply that theorem with $\Lambda=\Gamma$.
\end{proof}

We next state the circuit cost.  Let $n_H=\lceil\log_2N\rceil$, and
enlarge $b_{\rm val}$ if necessary so that it also bounds the bit
lengths of the public numerical parameters, including $t$, $\Lambda$,
$\Lambda_{\max}$, and $\eps$.  Assume that controlled
and inverse versions of $O_F,O_H$ are supplied as circuits of gate
costs $g_F,g_H$, using reusable workspaces of $a_F,a_H$ qubits.  We
also assume uniform reversible circuits for fixed-precision
arithmetic, comparisons, square roots, exponentials, trigonometric
functions, phases of complex numbers, and the implementation of
one-qubit rotations.
These circuits have size polynomial in their input word lengths and
the number of precision bits; let $a_{\rm ar}$ bound their reusable
workspace.
These assumptions do not provide constant-time access to a table of
clock amplitudes and do not assume qRAM.

\begin{theorem}[Near-linear gate realization]
\label{thm:gate}
Under the preceding gate-model assumptions and the hypotheses of
\cref{thm:main}, define the query-complexity scale
\begin{equation}
 Q_{\rm sim}:=t\sqrt d\,\Lambda
 +\sqrt{td\Lambda_{\max}\log\frac{2}{\eps}}
 +\log\frac{2}{\eps}.
\end{equation}
The simulator of \cref{thm:main} can be implemented using $O(Q_{\rm sim})$
sparse-oracle queries and
\begin{equation}\label{eq:gate-count}
 G_{\rm gates}=O\!\left(
 Q_{\rm sim}\left[g_F+g_H+
 \poly\!\left(
 n_H+b_{\rm val}+\log d+\log\frac{2Q_{\rm sim}}{\eps}
 \right)\right]
 \right)
\end{equation}
1- and 2-qubit gates.  The qubit count is
\begin{equation}\label{eq:qubit-count}
 O\!\left(
 n_H+b_{\rm val}+\log d+\log\frac{2Q_{\rm sim}}{\eps}
 +a_F+a_H+a_{\rm ar}
 \right).
\end{equation}
\end{theorem}

\subsection{Applications}

\paragraph{Quantum linear systems.}
Under the standard quantum linear systems problem (QLSP) access model,
sparse oracles are available for
both $A$ and $A^\dagger$, so Hermitian dilation handles a general
matrix without changing the asymptotic sparsity.  We may therefore let
$A$ be a Hermitian $d$-sparse matrix with $\norm{A}\le1$ and
$\norm{A^{-1}}\le\kappa$, and suppose that a preparation
circuit $U_b\ket0=\ket b$, together with its inverse and controlled
versions, is available.  Then
\cref{thm:main} implements $e^{-\iu A/2}$ with operator-norm error $\delta$
using $O(\sqrt{d\log(2/\delta)}+\log(2/\delta))$ sparse-oracle
queries, with $\Lambda=\Lambda_{\max}=1$.  Combining
this operation with the simulation-to-block-encoding conversion of
Low and Chuang and a block-encoded linear-system solver gives, with
constant success probability, an $\eps$-approximation to
\begin{equation}
 \frac{A^{-1}\ket{b}}{\norm{A^{-1}\ket{b}}}
\end{equation}
using
\begin{equation}
 \kappa\sqrt d\,\polylog(\kappa/\eps)
\end{equation}
queries \cite{LowChuangUniform,CostaEtAl2022}.  This replaces the
subpolynomial overhead in the earlier
Hamiltonian-simulation route of Low by polylogarithmic factors
\cite{Low2019}.  It does not, by itself, prove the optimal
single-logarithmic bound. Concurrent work of Bravo-Prieto, Harrow, and
Kothari establishes the worst-case asymptotic query complexity
\begin{equation}
 \Theta\!\left(\kappa\sqrt d\log\frac{1}{\eps}\right)
 \label{eq:intro-qls-optimal}
\end{equation}
by a different enlarged-linear-system construction
\cite{BravoPrietoHarrowKothari2026}.  As in that work, the matching
multiparameter lower bound is understood for sufficiently large
$\kappa$, $d$, and $1/\eps$, with the dimension allowed to grow. It is
not a uniform lower bound at degenerate parameter values.  Thus quantum linear systems
motivate the joint norm and sparsity dependence of our theorem, while
closing the remaining logarithmic gap would require an additional
end-to-end reduction.

\paragraph{Black-box unitary implementation.}
Here the reduction is direct.  Let $U$ be a unitary such that both $U$
and $U^\dagger$ have at most $d$ nonzero entries in each row, and
assume sparse access to both matrices.  Consider the Hermitian dilation
\begin{equation}
 H_U=\begin{pmatrix}0&U\\ U^\dagger&0\end{pmatrix}.
 \label{eq:intro-unitary-dilation}
\end{equation}
Every row and column of $H_U$ has Euclidean norm one, and
$H_U^2=\id$.  Therefore
\begin{equation}
 e^{-\iu(\pi/2)H_U}=-\iu H_U,
\end{equation}
so simulating $H_U$ for time $\pi/2$ implements $U$, up to a known
phase and a swap of the dilation qubit.  We obtain the following
corollary.

\begin{corollary}[Black-box unitary implementation]
\label{cor:black-box-unitary}
Given sparse-access oracles for a unitary $U$ and its adjoint, with
sufficient entry precision, and assuming that both $U$ and $U^\dagger$
are $d$-row-sparse, $U$ can be implemented with operator-norm error
$0<\eps\le1/2$ using
\begin{equation}
 O\!\left(\sqrt{d\log\frac{2}{\eps}}+\log\frac{2}{\eps}\right)
\end{equation}
queries.  If the supplied oracle circuits have gate cost $g_U$ and
use $a_U$ workspace qubits, then the implementation uses
\begin{equation}
 O\!\left(
 \left(\sqrt{d\log\frac{2}{\eps}}+\log\frac{2}{\eps}\right)
 \left[g_U+
 \poly\!\left(\log N+b_{\rm val}+\log d+\log\frac{2}{\eps}\right)
 \right]
 \right)
\end{equation}
1- and 2-qubit gates and a number of qubits obtained from
\cref{eq:qubit-count} by taking $a_F+a_H=O(a_U)$.
\end{corollary}

For an intended unitary whose entries are rounded, use consistent
oracles for the rounded matrix $\widetilde U$ and its adjoint, so that
the rounded dilation $\widetilde H_U$ is Hermitian.  If rounding
preserves the zero entries and every entry error is at most
$\zeta\le\eps/(\pi d)$, then
$\norm*{\widetilde H_U-H_U}=\norm*{\widetilde U-U}\le d\zeta$.
The evolution perturbation at time $\pi/2$ is therefore at most
$\eps/2$.  Moreover, the maximum row and column Euclidean norms of
$\widetilde H_U$ are at most $1+\sqrt d\,\zeta\le2$.
Every entry magnitude is also at most $2$.  Apply
\cref{thm:main,thm:gate} with $\Lambda=\Lambda_{\max}=2$ and simulation error
$\eps/2$ to obtain the same asymptotic bounds for the intended $U$.
In particular, $O(\log(d/\eps))$ fractional bits suffice, up to the
fixed encoding overhead.

For an arbitrary dense $N\times N$ unitary, the location oracle simply
enumerates all $N$ columns and the query complexity becomes
$O(\sqrt{N\log(2/\eps)}+\log(2/\eps))$.  At constant error this is
$\Theta(\sqrt N)$, matching the quantum-search lower bound.  Berry and
Childs asked whether this lower bound could be attained for every
black-box unitary; Low obtained $N^{1/2+o(1)}$ queries, and
\cref{cor:black-box-unitary} removes the subpolynomial factor
\cite{BerryChilds2012,Low2019}.  The same optimal constant-error query
bound was obtained independently in \cite{BravoPrietoHarrowKothari2026}.
The mixed-term simulation bound of \cite{Low2026MQSP} also implies
the precision-dependent estimate above.  
In addition to resolving the query question, \cref{thm:gate} supplies an
explicit gate bound under circuit access to the sparse location and
value oracles; for a dense unitary, the location circuit is trivial.

\subsection{Related work}

\paragraph{Sparse simulation with norm information.}
Low's 2019 algorithm decomposes a sparse Hamiltonian into entry-magnitude
bands, applies uniform spectral amplification to each band, and
recombines the resulting evolutions through an interaction-picture
construction \cite{Low2019,LowChuangUniform}.  This yields the
$1$-to-$2$ norm dependence in \cref{eq:intro-low-bound}, together
with an essentially query-linear gate overhead, but incurs
subpolynomial factors in $d$, $t$, and $1/\eps$.  Our improvement is
not a change from the spectral norm to the $1$-to-$2$ norm---Low's
theorem already has the latter, stronger dependence.  Our original
bound \cref{eq:original-query-bound}, established in the first arXiv
version of this paper \cite{LiWang2026v1}, removes the subpolynomial overhead
and, for $d>1$, gives an optimal linear query bound when
$t\Lambda\ge\log(2/\eps)$.  The current theorem applies to every
$t\ge0$, with the refined bound in \cref{eq:query-bound} obtained
by optimizing the working normalization.

\paragraph{Quantum-data-structure access.}
Chakraborty, Gily\'en, and Jeffery and, independently, Wang and Wossnig
developed Hamiltonian simulation algorithms for matrices stored in
quantum-accessible data structures
\cite{ChakrabortyGilyenJeffery2019,WangWossnig2020}.  Their data
structures prepare normalized rows or columns and lead to
useful Frobenius-norm or data-structure-dependent bounds, particularly
for non-sparse matrices.  This access model presupposes preprocessed
qRAM-like state preparation and is therefore incomparable with the
sparse location-and-value oracles in \cref{eq:OF,eq:OH}.  In
particular, it bypasses precisely the state-preparation cost whose
square-root sparsity dependence is central here.

\paragraph{Concurrent work.}
Low obtains the bound in \cref{eq:query-bound} through multivariate quantum signal
processing \cite[Theorem~6.4]{Low2026MQSP}.  Our original
fixed-normalization result \cite{LiWang2026v1} was obtained independently
and concurrently.
Following Low's normalization optimization, we apply
\cref{thm:abstract-log} with an adjusted working normalization to
obtain the mixed-term bound.  The sparse
lift and evolution transfer function agree with the corresponding
construction of \cite{Low2026MQSP} up to normalization and phase
conventions.  The distinction lies in the approximation proof and
finite circuit realization.  Our phase kernel exactly reproduces a
local Taylor approximation, expressed as a trigonometric polynomial,
reducing the simulation-error analysis to a localized remainder
estimate.  For the gate implementation, we construct an explicit
coupling from geometric endpoint states and a Fourier-diagonal
unitary, and control its discretization only on private states
reachable during the computation.  Together, these ingredients give
a self-contained constructive proof of the simulation bound through
local approximation and reachable-space analysis.  

Bravo-Prieto, Harrow, and Kothari prove the optimal QLSP bound
\cref{eq:intro-qls-optimal} by exactly block encoding a larger linear
system whose condition number is $O(\kappa\sqrt d)$
\cite{BravoPrietoHarrowKothari2026}.  As a separate consequence, they
construct a constant-normalization, error-$\eta$ block encoding of a
$d$-sparse Hermitian matrix of spectral norm at most one using
$O(\sqrt d\log(1/\eta))$ queries, and they implement black-box
unitaries with $O(\sqrt d\log(1/\eps))$ queries.  This gives the
same optimal constant-error complexity as
\cref{cor:black-box-unitary}; the mixed-term refinement improves the
precision dependence of our original implementation bound
\cite{LiWang2026v1}.  Their
starting point is the enlarged linear system, whereas ours is a
Hamiltonian-simulation theorem.  A generic composition of their block
encoding with qubitization does not give the refined bound in
\cref{eq:query-bound}; conversely, the generic conversion of our
simulator into a linear-system solver leaves the polylogarithmic
overhead described above.  Our theorem also uses the sharper
$1$-to-$2$ norm promise and includes circuit-level gate and qubit
accounting, while their stated results concern query complexity.

\paragraph{Transducers.}
The transducer framework of Belovs, Jeffery, and Yolcu describes a
unitary that transforms a public input while returning a private
catalyst unchanged, and it supports exact and thrifty composition of
quantum subroutines \cite{BelovsJefferyYolcu2024}.  Recent algorithms
for time-dependent Hamiltonian simulation use this framework to share
oracle calls across many labeled private sectors
\cite{ChenGaoWangZhou2026,ChenEtAlGate2026}.  We adapt the continuous
coupling of \cite[Appendix~A.1]{ChenGaoWangZhou2026} and the shared-private-space
reuse arrangement of \cite[Theorem~3.2]{BelovsJefferyYolcu2024} to
the sparse-entry Cayley query.  Relative to these transducer
constructions, our analysis additionally uses the compressed
second-moment bound, localized signed phase weights, and
different input and output clock states.  We analyze the routed circuit
directly from zero private input, without preparing an infinite limiting
catalyst or invoking a generic transducer implementation theorem.
Our gate realization uses an explicit Fourier-diagonal coupling rather
than the dyadic factorization in \cite{ChenEtAlGate2026}.

\subsection{Technical overview}
\label{sec:technical-overview}

We now outline the main ideas behind the query bound and its gate
implementation.  The starting point is to express Hamiltonian
evolution as an effective input--output transformation arising from
unitary interactions with an auxiliary space.  A finite routing clock
then approximates this transformation, and a structured circuit
implements the resulting construction efficiently.

Three spaces play complementary roles.  The \emph{public space}
$\HH$ holds the system state at input and output.  The \emph{clock
space} $\bbc_C^R$ labels $R$ public interaction slots, each carrying a
copy of the system space.  The \emph{private space} $\QQ$ temporarily
carries amplitudes between public interactions and is where the
matrix-dependent queries act.  The clock is attached to the public
slots; a single private space is shared by all of them, rather than
being copied for each clock value.

The computation starts with the public input distributed over a
superposition of clock slots and no component in the private space.
Successive interactions couple one public slot at a time to the shared
private space.  Amplitudes can enter from one slot, remain private
during subsequent queries, and return at a later slot.  Choosing the
input and output clock states determines how these contributions are
combined.  The clock thus labels the interaction schedule; it is not
advanced by a separate shift operation.  The steps below explain how
we choose the interaction and clock states to approximate
$e^{-\iu tH}$ with the claimed resource bounds.

\paragraph{Sparse-entry lifting.}
The first step is to capture the $1$-to-$2$ norm information in a
representation that can be queried directly.  We embed the system by
an input-independent isometry $V$ that prepares a uniform list-index
state, and construct a Hermitian operator $M$ from one- and
two-dimensional blocks of matrix entries.  \cref{lem:entry-lifting}
gives
\[
 H=2\sqrt d\,\Lambda V^\dagger MV,
 \qquad V^\dagger M^2V\preceq\frac14\id.
\]
The second bound is the useful consequence of the norm promise:
$\norm{MV}\le1/2$, even though $\norm{M}$ can be of order $\sqrt d$.
We access $M$ through its Cayley transform
$U_M=(M+\iu\id/2)(M-\iu\id/2)^{-1}$.  Its small blocks allow each
application to be implemented with $O(1)$ sparse-oracle queries.

\paragraph{Evolution as a transfer response.}
We next combine the Cayley query with a unitary coupling that depends
only on public parameters and exchanges amplitudes between the public
system and a private history space.
Feeding the private output back through $zU_M$ defines a transfer
function $\Phi_t(z)$.  \cref{prop:transfer} identifies its boundary
values with a family of Hamiltonian evolutions $e^{-\iu tH_\theta}$,
where $H_0=H$; in particular,
\[
 \Phi_{t,\partial}(1)=e^{-\iu tH}.
\]
Simulation therefore reduces to recovering the phase-zero value of
this boundary family.  The feedback is an analytical representation
of the target, not an infinite sequence of operations executed by the
algorithm.

\paragraph{Routed reuse of the private space.}
A routing clock realizes weighted transfer responses in a circuit.
Successive public slots interact with one shared private space.  The
input is supported only on the public slots, so the initial private
component is zero.  An amplitude entering from slot
$j$ and spending $\ell$ query steps in the private space returns to
slot $j+\ell$.  Consequently, correlations between the input and output
clock amplitudes determine the coefficients associated with different
private residence times.  \cref{lem:cross-spectral-reuse} expresses
the resulting input--output block as a weighted integral of
$e^{-\iu tH_\theta}$.  This integral describes the block implemented by
clock preparation, $R$ routed interactions, and clock unpreparation;
the algorithm does not sample the phase $\theta$.

\paragraph{Localization and exact cancellation.}
To recover the value at zero, we construct real, even phase weights
$\kappa_{m,p}$ with two complementary properties.  They are localized
near zero and exactly reproduce Fourier modes of order at most $p$:
$\int_{-\pi}^{\pi}\kappa_{m,p}(\theta)e^{\iu k\theta}\,d\theta=1$
for $|k|\le p$.  Thus their integral against any trigonometric
polynomial of degree at most $p$ equals that polynomial's value at
zero.  After pairing opposite phases, the truncated local expansion
of the boundary response, written in powers of
$\sin^2(\theta/2)$, is such a polynomial and is reproduced exactly.
Localization controls the remaining error, both within the
expansion neighborhood and away from it.  The compressed second-moment
bound controls the growth of the response in that neighborhood.
Let $L_0\ge\max\{1,\norm{M}\}$ control the inverse width of the
analytic neighborhood.  \cref{prop:filtered-error} gives an error
of order $e^{\tau/(2L_0)}16^{-p}$, so we can choose
\[
 m=\Theta(L_0),
 \qquad
 p=\Theta\!\left(\frac{\tau}{L_0}+\log\frac{2}{\eps}\right).
\]
The required number $R$ of public clock slots, equal to the number of
routed interactions, is
\[
 R=O(mp)
  =O\!\left(\tau+L_0\log\frac{2}{\eps}\right).
\]

\paragraph{A phase-twisted clock and amplitude amplification.}
The phase weights are signed, whereas using the same input and output
clock state produces only nonnegative phase weights.  We instead use a
phase-twisted pair of different clock states.  \cref{lem:single-clock}
factors the signed weights into their cross-spectrum and realizes the
phase integral with one routed circuit and a known normalization
$1\le\lambda\le2$.  One round of oblivious amplitude
amplification (OAA) then converts this block into a simulator with
operator-norm error $\eps$, including the output ancillas as in
\cref{eq:simulation-error}.  By \cref{lem:purification}, the round
uses only three calls to the block circuit or its adjoint, preserving
the query bound above.

\paragraph{Optimizing the working normalization.}
The given norm bound need not be the normalization used by the
construction.  Running the same construction with any
$\Lambda'\ge\Lambda$ preserves the compressed second-moment bound,
while changing the evolution scale to $\tau'=t\sqrt d\,\Lambda'$
and permitting
$L_0=\max\{1,\sqrt d\,\Lambda_{\max}/(2\Lambda')\}$.
Increasing $\Lambda'$ therefore increases the evolution contribution
but decreases the precision contribution.  For $t>0$, choose
\[
 \Lambda'=\max\left\{
 \Lambda,\sqrt{\frac{\Lambda_{\max}\log(2/\eps)}{t}},
 \frac{1}{t\sqrt d}\right\}.
\]
The first two terms balance these costs; the last ensures
$\tau'\ge1$ for the finite-gate analysis and adds only a constant
to the query count.  Substitution in the preceding bound gives
\cref{eq:query-bound}.  At $t=0$, the identity circuit suffices.
This optimization changes only the normalization and clock
parameters, not the kernel or amplification construction.

\paragraph{Gate implementation.}
The query construction specifies a unitary coupling, but an arbitrary
unitary realization need not have a low gate cost.  The final step is
therefore to replace the analytical coupling by an explicit structured
circuit on a discretized history register.  The circuit first encodes
the public amplitude into a state with geometrically decaying
amplitudes, applies a quantum Fourier transform, applies a phase
computed from the frequency address, and then uses the inverse
transform and output decoding.  The encoding amplitudes and frequency
phases have explicit formulas, so these operations do not require a
table with one entry per grid point or access to qRAM.  This structure
gives a gate cost per coupling that is polynomial in the address length
and working precision, rather than in the number of grid points.

The same private history register is reused throughout the $R$ routed
interactions.  To control the discretization error, we only need to
compare the circuit with the analytical coupling on histories that
can actually arise from the initially empty private space.  These
histories have a controlled analytic form, allowing
\cref{lem:coupling-comparison,lem:forward-grid-hybrid} to bound both the
error of one interaction and its accumulation over the computation.
A grid polynomially large in $\tau'$, $R$, and $1/\eps$ suffices and
requires only a logarithmic-size address register.  Importantly, the
small address register alone does not guarantee an efficient circuit;
the Fourier structure and explicit encodings above also ensure that
we need not process the grid points one by one.

The routing-clock preparations have a different cost: each uses
$O(R)$ controlled rotations and phases.  They occur only constantly
many times, before or after a full routed computation, and are not
repeated at every interaction.  Their total cost is therefore
near-linear in $R$, rather than quadratic.  Finally, we include the
gate costs of the supplied sparse-oracle circuits and the reversible
arithmetic, and choose the working precision so that errors from
finite-precision arithmetic and approximate gates remain within the
simulation error budget.  Under the stated
gate-model assumptions, this gives a total gate count linear in the
query-complexity scale, multiplied by the oracle costs and a polynomial
overhead in the input bit lengths and logarithmic precision
parameters, as stated in \cref{thm:gate}.

\subsection{Organization}
\cref{sec:preliminaries} reviews projected-unitary encodings and the
one-round OAA step.  \cref{sec:routed-construction} develops the
routed unitary construction: the sparse-entry representation and
Cayley query, the unitary input--output realization of Hamiltonian
evolution, and finite reuse of a shared private space.
\cref{sec:logarithmic} proves the logarithmic-precision approximation,
and \cref{sec:gate-realization} gives the finite gate implementation
and completes the resource accounting.
The OAA proof appears in \cref{app:purification}.

\subsection{AI Disclosure}
The core technical results of this paper were discovered through
discussions with GPT-5.6 Sol and GPT-6 Astra.  The authors verified and
simplified the proofs and take sole responsibility for the correctness
of the final manuscript.  GPT-5.6 Sol and GPT-6.1 Sol assisted with
writing the manuscript.

\section{Preliminaries}
\label{sec:preliminaries}

\paragraph{Notation and conventions.}
All Hilbert spaces are complex.  For finite-dimensional vectors,
$\norm{v}$ denotes the Euclidean norm,
\[
 \norm{v}=\left(\sum_j|v_j|^2\right)^{1/2}.
\]
Function-space norms are specified locally.
For a bounded linear operator $T$, we write
\[
 \norm{T}:=\sup_{\norm{v}=1}\norm{Tv}
\]
for the induced operator norm.  For matrices, this is the spectral
norm.  All operator-error estimates,
including those between input--output isometries, use this norm.
For a matrix $A$ in the specified computational basis, define
\[
 \norm{A}_{1\to2}
 :=\max_j\left(\sum_i|A_{ij}|^2\right)^{1/2},
 \qquad
 \norm{A}_{\max}:=\max_{i,j}|A_{ij}|.
\]
Thus the $1$-to-$2$ norm is the maximum column Euclidean norm; for
Hermitian matrices it is also the maximum row Euclidean norm.

For a complex number $c$, we use $\overline{c}$ to denote its complex
conjugate.  We write $T^\dagger$ for the adjoint and $\id$ for
the identity on the space understood from context.  Projectors are orthogonal, and for
Hermitian operators $A,B$ on the same space, $A\preceq B$ means that
$B-A$ is positive semidefinite.  For Hilbert spaces, vectors, and
operators, $\oplus$ denotes the orthogonal direct sum, while
$\otimes$ denotes the tensor product.  In particular,
\[
 \norm{x\oplus y}^2=\norm{x}^2+\norm{y}^2.
\]
In the value oracle \cref{eq:OH}, $\oplus$ instead denotes bitwise XOR.

We write $[n]=\{1,\ldots,n\}$ and $\iu$ for the imaginary unit.
Logarithms are natural unless a base is displayed.  Implicit constants
in asymptotic estimates are absolute unless indicated otherwise;
$a\lesssim b$ means $a\le Cb$ for an absolute constant $C$.
The notation $\widetilde O$ suppresses polylogarithmic factors in the
relevant parameters, while $\poly(x)$ and $\polylog(x)$ denote
unspecified polynomials in $x$ and $\log x$, respectively.

\paragraph{Projected-unitary encodings.}
If $\Pi$ and $\widetilde\Pi$ are projectors whose ranges are
identified with the input and output system spaces, a
projected-unitary encoding of $Q$ is a unitary $\mathcal U$ such
that $\widetilde\Pi\mathcal U\Pi=Q$.  The special case
$Q=\bra{0}\mathcal U\ket{0}$ is called a clean block encoding.  We
use the standard calculus of projected unitary encodings
\cite{Gilyen2019}.

A clean block encoding of an operator close to a unitary does not
directly implement that unitary.  Oblivious amplitude amplification
was introduced in \cite{BerryEtAl2014}, building on the two-reflection
mechanism of Grover's search algorithm \cite{Grover1996} and the
in-place amplification technique of Marriott and Watrous
\cite{MarriottWatrous2005}.  Its extension to operators close to
unitaries was developed in \cite{BerryEtAl2015}.
The following lemma uses one round to convert
a block $Q/\lambda$, with known $1\le\lambda\le2$, into a unitary
implementation whose operator-norm error is linear in $\norm{Q-U}$.
The round is realized by one application of a pair of reflections.

\begin{lemma}[One-round oblivious amplitude amplification]
\label{lem:purification}
Suppose $U$ is unitary, $\norm{Q-U}\le\delta<1/4$, and
$Q/\lambda$ is a clean block of a circuit with known
$1\le\lambda\le2$.  Three calls to that circuit or its adjoint
suffice to construct a unitary $\mathcal C$ on the joint system and
ancillary workspace such that
\begin{equation}\label{eq:purification}
 \norm{\mathcal C(\ket{0}_{\rm anc}\otimes\id)
       -\ket{0}_{\rm anc}\otimes U}
 \le4\delta.
\end{equation}
Here $\ket{0}_{\rm anc}$ denotes the all-zero state of the
block-encoding ancillas together with the additional qubit introduced
in the proof.
\end{lemma}

The proof is given in \cref{app:purification}.

\section{A routed unitary construction for Hamiltonian simulation}
\label{sec:routed-construction}

\subsection{Sparse-entry representation and the Cayley query}
\label{sec:entry}

The construction in this section is valid for any working normalization
$\Lambda\ge\norm{H}_{1\to2}$.  We use the separate entry bound
$\Lambda_{\max}\ge\norm{H}_{\max}$ when bounding the lifted operator.
In the proof of \cref{thm:main}, we apply the construction with an
enlarged normalization $\Lambda'$ to balance the time-dependent and
precision-dependent costs.

Let
\[
 \EE:=\Span\{\ket{i,\ell}:i,\ell\in[N]\},
 \qquad
 \EE_{\rm list}:=
 \Span\{\ket{i,\ell}:i\in[N],\ 1\le\ell\le d\}.
\]
We regard $\EE$ as the ambient row--list-index workspace and
$\EE_{\rm list}$ as its sparse-list subspace.
Since each $f_i$ is a permutation of $[N]$, $O_F$ is unitary on
$\EE$.  Its image
\[
 \widehat\EE_{\rm list}:=O_F\EE_{\rm list}
\]
is the listed-address subspace in row--column coordinates.  The
registers in $\EE$ store addresses only; $O_H$ loads the numerical
value $H_{ij}$ separately into the value register.  Define the
input-independent isometry $V:\HH\to\EE_{\rm list}$ by
\begin{equation}\label{eq:V}
 V\ket{i}=\frac{1}{\sqrt d}\sum_{\ell=1}^{d}\ket{i,\ell}.
\end{equation}
For every $i<j$ with $H_{ij}\ne0$, both transpose-related address states
$\ket{i,j}$ and $\ket{j,i}$ belong to $\widehat\EE_{\rm list}$, and we
set
\begin{equation}\label{eq:mij}
 \mu_{ij}=\frac{\sqrt d}{2\Lambda}H_{ij},
 \qquad
 \widehat M\big|_{\Span\{\ket{i,j},\ket{j,i}\}}
 =\begin{pmatrix}0&\mu_{ij}\\ \overline{\mu_{ij}}&0\end{pmatrix}.
\end{equation}
For a nonzero diagonal entry, $\widehat M$ acts on $\ket{i,i}$ by the
real scalar $\sqrt d H_{ii}/(2\Lambda)$.  Set $\widehat M$ to zero on
all remaining basis states of $\widehat\EE_{\rm list}$, including padded
entries, and extend it by zero on $\widehat\EE_{\rm list}^{\perp}$.
Finally, define the operator on the ambient address workspace $\EE$ by
\[
 M:=O_F^\dagger\widehat M O_F.
\]

The Cayley query acts most naturally on transpose pairs of listed
row--column addresses,
whereas the desired evolution is expressed on the original system
space.  The promise in~\cref{eq:one-two-promise} is exactly what is
needed after compression: squaring a transpose-pair block and averaging
over list indices produces a normalized squared row Euclidean norm.  Hermiticity
identifies row and column Euclidean norms, so the resulting compressed
second moment is bounded even though the uncompressed operator $M$ can
have norm of order $\sqrt d$.

\begin{lemma}[Sparse-entry lifting and compressed second moment]
\label{lem:entry-lifting}
The operators $V$ and $M$ satisfy
\begin{equation}\label{eq:compression}
 H=2\sqrt d\,\Lambda V^\dagger M V,
 \qquad
 V^\dagger M^2V\preceq\frac14\id.
\end{equation}
Moreover,
\begin{equation}\label{eq:Mnorm}
 \norm{M}
 =\frac{\sqrt d}{2\Lambda}\max_{i,j\in[N]}|H_{ij}|
 \le\frac{\sqrt d\,\Lambda_{\max}}{2\Lambda}
 \le\frac{\sqrt d}{2}.
\end{equation}
\end{lemma}

\begin{proof}
The two factors $d^{-1/2}$ in~\cref{eq:V} give, for every
$i\ne j$ with $H_{ij}\ne0$,
\[
 \bra{i}V^\dagger MV\ket{j}
 =\frac{H_{ij}}{2\sqrt d\,\Lambda}.
\]
The diagonal convention gives the same identity on diagonal entries,
proving the first part of~\cref{eq:compression}.  Squaring each
$2\times2$ transpose-pair block gives $|\mu_{ij}|^2\id$.  Hence
\begin{align}
 V^\dagger M^2V
 =\sum_i\ket{i}\!\bra{i}\,
   \frac{1}{4\Lambda^2}\sum_j|H_{ij}|^2 \\
 =\sum_i\ket{i}\!\bra{i}\,
   \frac{1}{4\Lambda^2}\sum_j|H_{ji}|^2
 \preceq\frac14\id,
\end{align}
where the second line uses Hermiticity and the last inequality uses
\cref{eq:one-two-promise}.  Finally, the block definition and
$\max_{i,j\in[N]}|H_{ij}|\le\Lambda_{\max}\le\Lambda$ give
\cref{eq:Mnorm}.  Notice that the compressed moment,
rather than a constant bound on $\norm{M}$, is the useful estimate.
\end{proof}

\begin{remark}[Why the $1$-to-$2$ norm is the natural promise]
\label{rem:one-two-information}
The sparse-entry lifting argument uses the input normalization only through the
row Euclidean bounds.  In particular, every entry obeys
$|H_{ij}|\le\Lambda$.  For Hermitian $H$, both facts follow from
$\norm{H}_{1\to2}\le\Lambda$; no spectral-norm promise is needed.
This is precisely the column-Euclidean norm parameter appearing in
Low's sparse-simulation bounds \cite{Low2019}.  Since
$\norm{H}_{1\to2}\le\norm{H}$, a spectral-norm bound is a stronger
promise and yields \cref{cor:spectral-norm} as an immediate
specialization.
\end{remark}

Define the address-space Cayley unitary
\begin{equation}\label{eq:cayley}
 U_M:=\operatorname{Cay}(M)
 =(M+\iu\id/2)(M-\iu\id/2)^{-1}.
\end{equation}
For a scalar spectral value $\lambda_M\in\bbr$, the corresponding
eigenvalue $(\lambda_M+\iu/2)/(\lambda_M-\iu/2)$ has modulus one.  On a
transpose-pair block
$M_\mu=\left(\begin{smallmatrix}0&\mu\\\overline\mu&0\end{smallmatrix}\right)$,
\begin{equation}\label{eq:cayley-block}
 \operatorname{Cay}(M_\mu)
 =\frac{(|\mu|^2-1/4)\id+\iu M_\mu}{|\mu|^2+1/4}.
\end{equation}
Thus one application of $U_M$ first converts list-index coordinates to
row--column addresses.  For $i\ne j$, reversibly encode the address as
the canonically ordered pair $(\min\{i,j\},\max\{i,j\})$ together with
an orientation bit, query $H_{\min\{i,j\},\max\{i,j\}}$, and apply
\cref{eq:cayley-block} to the orientation bit.  The canonical pair is
unchanged by this rotation, so the queried value can be uncomputed
before undoing the ordering.  This is necessary for complex Hermitian
entries: querying $H_{ij}$ and then directly mixing the two addresses
would not allow uncomputation, since $H_{ji}=\overline{H_{ij}}$.
Diagonal addresses require only the corresponding scalar phase.
Finally, return to list-index coordinates.  We call one use of $U_M$ or
$U_M^\dagger$ a \emph{Cayley query}; each such use costs $O(1)$ calls
to the original sparse oracles.

\subsection{A unitary input--output realization of Hamiltonian evolution}
\label{sec:input-output-realization}

We adapt the continuous coupling of
\cite[Appendix~A.1]{ChenGaoWangZhou2026} to the embedding $V$ and
the scale $\sqrt d\,\Lambda$, with a different phase convention on
the private output.  We then couple it to the sparse-entry Cayley
query to obtain the transfer family analyzed below.

Introduce the continuous private space
\begin{equation}\label{eq:history}
 \QQ=L^2([0,t];\EE).
\end{equation}
Every operator $T$ on $\EE$ is identified with its pointwise lift to
$\QQ$:
\[
 (Tq)(r):=T(q(r))
 \qquad\text{for almost every }r\in[0,t].
\]
In particular, this convention applies to $\Pi_V:=VV^\dagger$, $M$,
and $U_M$.

Set $\gamma=\sqrt{2\sqrt d\,\Lambda}$.  For
$(\xi,q)\in\HH\oplus\QQ$, define
\begin{align}
 \psi(r)&=e^{-\sqrt d\,\Lambda r}\xi
   -\iu\gamma\int_0^r e^{-\sqrt d\,\Lambda(r-u)}V^\dagger q(u)\,du,
   \label{eq:psi-def}\\
 g(r)&=q(r)-\iu\gamma V\psi(r),
 \qquad
 W_t(\xi,q)=(\psi(t),g).
 \label{eq:Wt-def}
\end{align}
Equivalently,
\begin{equation}\label{eq:Wt-block}
 W_t=\begin{pmatrix}D&C_0\\B&A_0\end{pmatrix},
 \qquad D=e^{-t\sqrt d\,\Lambda}\id,
\end{equation}
where
\begin{align}
 C_0q&=-\iu\gamma\int_0^t e^{-\sqrt d\,\Lambda(t-r)}V^\dagger q(r)\,dr,
 \label{eq:C0}\\
 (B\xi)(r)&=-\iu\gamma e^{-\sqrt d\,\Lambda r}V\xi,
 \label{eq:B}\\
 (A_0q)(r)&=q(r)-2\sqrt d\,\Lambda\int_0^r
 e^{-\sqrt d\,\Lambda(r-u)}\Pi_Vq(u)\,du.
 \label{eq:A0}
\end{align}

To realize Hamiltonian evolution as a transfer response, we first
separate the matrix-dependent Cayley query from a coupling that depends
only on public parameters.  The following unitary realization supplies
that fixed public--private interaction.  Its unitarity ensures that
every subsequent routed reuse is a physical operation that stores and
returns private history without loss.

\begin{lemma}[Unitary coupling]\label{lem:Wt}
$W_t$ is unitary and depends only on the public parameters
$d,\Lambda,t,V$, not on the queried values of $H$.
\end{lemma}

\begin{proof}
Differentiating~\cref{eq:psi-def} gives
$\dot\psi=-\sqrt d\,\Lambda\psi-\iu\gamma V^\dagger q$.  Using
$V^\dagger V=\id$ and $\gamma^2=2\sqrt d\,\Lambda$,
\begin{equation}\label{eq:energy-balance}
 \frac{d}{dr}\norm{\psi(r)}^2
 =\norm{q(r)}^2-\norm{g(r)}^2.
\end{equation}
Integration from $0$ to $t$ proves norm preservation.  For
surjectivity, prescribe $(\xi_{\rm out},g)$, solve
$\dot\psi=\sqrt d\,\Lambda\psi-\iu\gamma V^\dagger g$ backward from
$\psi(t)=\xi_{\rm out}$, and set $q=g+\iu\gamma V\psi$.  This constructs
the inverse.
\end{proof}

For a scalar parameter $z$, apply the Cayley query $U_M$ and the factor
$z$ on the private sector, followed by the fixed coupling:
\begin{equation}\label{eq:S}
 S(z):=W_t\diag(\id,zU_M)
 =\begin{pmatrix}D&zC\\B&zA\end{pmatrix},
 \qquad C=C_0U_M,\quad A=A_0U_M.
\end{equation}
For $z=e^{\iu\theta}$, write $S_\theta:=S(e^{\iu\theta})$; this is
unitary because both factors in \cref{eq:S} are unitary.

The purpose of this input--output realization becomes clear after its
closed-loop response is identified with Hamiltonian evolution.  The
next proposition shows that the feedback parameter $z$ produces an explicit
rational deformation of $H$, with the desired evolution appearing at
$z=1$.  This is the bridge that reduces Hamiltonian simulation to
recovering the value at $\theta=0$ from the boundary family
$\theta\mapsto e^{-\iu tH_\theta}$.

\begin{proposition}[Evolution transfer function]
\label{prop:transfer}
For $|z|<1$, the transfer function
\begin{equation}\label{eq:Phi}
 \Phi_t(z):=D+zC(\id-zA)^{-1}B
\end{equation}
satisfies
\begin{equation}\label{eq:Phi-exp}
 \Phi_t(z)=\exp\!\left[-t\sqrt d\,\Lambda V^\dagger
 (\id+zU_M)(\id-zU_M)^{-1}V\right].
\end{equation}
It is analytic and contractive in the disk.  Whenever the radial limit
exists, write
$\Phi_{t,\partial}(e^{\iu\theta})
:=\lim_{\varrho\uparrow1}\Phi_t(\varrho e^{\iu\theta})$.
Apart from a finite set of phases, this boundary value is
\begin{equation}\label{eq:boundary}
 \Phi_{t,\partial}(e^{\iu\theta})=e^{-\iu tH_\theta},
 \qquad
 H_\theta=\sqrt d\,\Lambda V^\dagger
 (2M-u\id)(\id+2uM)^{-1}V,
 \quad u=\tan(\theta/2).
\end{equation}
In particular,
\begin{equation}\label{eq:target-boundary}
 \Phi_{t,\partial}(1)=e^{-\iu tH}.
\end{equation}
\end{proposition}

\begin{proof}
Feed the private output $g$ back as the private input $zg$ before the
Cayley query.  The field entering $W_t$ is then $q=zU_Mg$.  Substituting
$g=q-\iu\gamma V\psi$ gives
\begin{equation}
 q=-\iu\gamma(\id-zU_M)^{-1}zU_MV\psi.
\end{equation}
Consequently
\begin{equation}
 \dot\psi=-\sqrt d\,\Lambda V^\dagger
 (\id+zU_M)(\id-zU_M)^{-1}V\psi,
\end{equation}
which proves~\cref{eq:Phi-exp}.  Unitarity with feedback input $zg$
also yields
\begin{equation}
 \norm{\Phi_t(z)\xi}^2
 =\norm{\xi}^2-(1-|z|^2)\norm{g}^2,
\end{equation}
and hence contractivity.

For a real spectral value $\lambda_M$,
\begin{equation}
 \frac{1+e^{\iu\theta}(\lambda_M+\iu/2)/(\lambda_M-\iu/2)}
      {1-e^{\iu\theta}(\lambda_M+\iu/2)/(\lambda_M-\iu/2)}
 =\iu\frac{2\lambda_M-u}{1+2u\lambda_M}.
\end{equation}
Functional calculus proves~\cref{eq:boundary}.  At $\theta=0$,
$(\id+U_M)(\id-U_M)^{-1}=2\iu M$, so
\cref{eq:compression} gives~\cref{eq:target-boundary}.
\end{proof}

\begin{remark}[Transducer interpretation]
For $0<\varrho<1$, the resolvent vector
\[
 v_{\varrho,\xi}=(\id-\varrho A)^{-1}B\xi
\]
is well defined and obeys
\[
 S(1)(\xi\oplus\varrho v_{\varrho,\xi})
 =\Phi_t(\varrho)\xi\oplus v_{\varrho,\xi}.
\]
This regulated input--output relation is the sense in which the
private history resembles a transducer catalyst space.  We do not
assert that $v_{\varrho,\xi}$ converges in the infinite history norm as
$\varrho\uparrow1$, or that $\id-A$ has a bounded inverse there.  The
algorithm does not prepare a limiting catalyst and does not invoke a
generic transducer implementation theorem; it starts with zero private
amplitude and compiles the transfer response directly with a finite
clock.
\end{remark}

\subsection{Finite routed reuse of a shared private space}
\label{sec:finite-clock}

We use the shared-private-space reuse arrangement underlying
\cite[Theorem~3.2]{BelovsJefferyYolcu2024}; see also
\cite[Section~5]{ChenGaoWangZhou2026}.  Here we allow independently
chosen input and output clock states and analyze the resulting
weighted transfer response directly.

Let $R\in\mathbb N$ with $R\ge1$, and set
$\PP_R=\bigoplus_{k=0}^{R-1}\HH_k\cong\bbc_C^R\otimes\HH$.
We call $\bbc_C^R$ the \emph{routing clock register}, with basis
$\{\ket{0}_C,\ldots,\ket{R-1}_C\}$.  Its basis label records the
scheduled public interaction slot; it is not advanced by an
autonomous shift operator.  All public slots share the same private
space $\QQ$.  For a normalized clock state
$\ket{\mathbf a}_C=\sum_{j=0}^{R-1}a_j\ket{j}_C$, with clock-state
amplitude vector $\mathbf a=(a_0,\ldots,a_{R-1})$, define
\begin{equation}\label{eq:Ja}
 J_{\mathbf a}\xi=(a_0\xi,\ldots,a_{R-1}\xi;0)
 \in\PP_R\oplus\QQ.
\end{equation}
For $z\in\bbc$ and each $k$, let $S^{(k)}(z)$ act on
$\HH_k\oplus\QQ$ through the block matrix
$\left(\begin{smallmatrix}D&zC\\B&zA\end{smallmatrix}\right)$ and as
the identity on the other public slots:
\begin{equation}\label{eq:Sk-action}
 S^{(k)}(z):\quad
 \binom{x_k}{q}\longmapsto
 \binom{Dx_k+zCq}{Bx_k+zAq},
 \qquad x_j\longmapsto x_j\quad(j\ne k).
\end{equation}
On the unit circle, write
$S_\theta^{(k)}=S^{(k)}(e^{\iu\theta})$; this is a routed unitary
direct-sum interaction, not a clock increment and not $S_\theta^R$ on
one public--private pair.  Equivalently, if $\mathsf{Swap}_{0,k}$
swaps public slots $0$ and $k$, then
$S^{(k)}(z)=\mathsf{Swap}_{0,k}^\dagger
S^{(0)}(z)\mathsf{Swap}_{0,k}$.

Define the complete length-$R$ routed circuit by
\begin{equation}\label{eq:routed-product}
 \mathcal R_R(z)
 :=S^{(R-1)}(z)\cdots S^{(0)}(z).
\end{equation}
For $z=e^{\iu\theta}$, this is a unitary circuit and equals
$S_\theta^{(R-1)}\cdots S_\theta^{(0)}$.
If $g_{k-1}$ is the private amplitude just before slot $k$, with
$g_{-1}=0$, then
\begin{equation}\label{eq:clock-recurrence}
 y_k=a_kD\xi+zCg_{k-1},
 \qquad
 g_k=a_kB\xi+zAg_{k-1}.
\end{equation}
Solving the recurrence gives
\begin{align}
 g_k&=\sum_{j=0}^k a_jz^{k-j}A^{k-j}B\xi,\label{eq:gk}\\
 y_k&=a_kD\xi+\sum_{\ell=1}^k
       a_{k-\ell}z^\ell CA^{\ell-1}B\xi.
 \label{eq:yk}
\end{align}
Thus an amplitude entering from slot $j$ through $B$, remaining private
for $\ell-1$ further interactions through $A$, and exiting through $C$
during step $j+\ell$ contributes
$a_jz^\ell\ket{j+\ell}CA^{\ell-1}B\xi$.

The exact transfer function permits arbitrarily many private returns,
but an algorithm can make only finitely many Cayley queries.  We use
one clock state to inject the public amplitude and may use a different
clock state to collect it.  Each private residence time then becomes a
cross-correlation coefficient.  This bilinear form is essential below:
it permits a signed phase kernel without an LCU over multiple routed
circuits.

For normalized clock-state amplitude vectors
$\mathbf a,\mathbf b\in\bbc^R$, let
$\ket{\mathbf a}_C=\sum_ja_j\ket{j}_C$ and
$\ket{\mathbf b}_C=\sum_jb_j\ket{j}_C$ be the corresponding states.
Define
their prepare--route--unprepare block by
\begin{equation}\label{eq:Pa}
 P_{\mathbf b\leftarrow\mathbf a}(z)
 :=J_{\mathbf b}^\dagger\mathcal R_R(z)J_{\mathbf a}.
\end{equation}
Associate with $\mathbf a$ and $\mathbf b$ the generating polynomials
\[
 a(z)=\sum_{j=0}^{R-1}a_jz^j,
 \qquad
 b(z)=\sum_{j=0}^{R-1}b_jz^j,
\]
and define their clock cross-spectrum by
\begin{equation}\label{eq:rhoa}
 \rho_{\mathbf b\leftarrow\mathbf a}(\theta)
 :=
 \frac{1}{2\pi}
 \overline{b(e^{\iu\theta})}a(e^{\iu\theta}).
\end{equation}

The clock correlations appearing in the routed circuit are precisely
the Fourier coefficients of this cross-spectrum.  The following lemma
converts the finite routed circuit into the boundary-phase integral
needed by the phase-twisted clock construction.  For general
$\mathbf a$ and $\mathbf b$, the cross-spectrum may be complex and has
integral $\langle\mathbf b,\mathbf a\rangle$, so the integral need not
be an average.

\begin{lemma}[Finite-clock cross-spectral representation]
\label{lem:cross-spectral-reuse}
For normalized $\mathbf a,\mathbf b\in\bbc^R$, the block and clock
cross-spectrum in
\cref{eq:Pa,eq:rhoa} satisfy
\begin{equation}\label{eq:boundary-phase-integral}
 P_{\mathbf b\leftarrow\mathbf a}(1)
 =
 \int_{-\pi}^{\pi}
 \rho_{\mathbf b\leftarrow\mathbf a}(\theta)
 e^{-\iu tH_\theta}\,d\theta.
\end{equation}
\end{lemma}

\begin{proof}
Substituting \cref{eq:yk} into \cref{eq:Pa} and projecting the public
output against $\sum_k\overline{b_k}\bra{k}$ gives
\[
 P_{\mathbf b\leftarrow\mathbf a}(z)
 =
 \langle\mathbf b,\mathbf a\rangle D
 +
 \sum_{\ell=1}^{R-1}
 \left(
  \sum_{j=0}^{R-1-\ell}
  \overline{b_{j+\ell}}a_j
 \right)
 z^\ell CA^{\ell-1}B.
\]
The parenthesized coefficients, together with the zeroth coefficient
$\langle\mathbf b,\mathbf a\rangle$, are precisely the Fourier
coefficients of $\rho_{\mathbf b\leftarrow\mathbf a}$ paired with the
corresponding return delays.  Hence the norm-convergent return expansion
of $\Phi_t$ gives, for
$0<\varrho<1$,
\[
 P_{\mathbf b\leftarrow\mathbf a}(\varrho)
 =
 \int_{-\pi}^{\pi}
 \rho_{\mathbf b\leftarrow\mathbf a}(\theta)
 \Phi_t(\varrho e^{\iu\theta})\,d\theta.
\]
As $\varrho\uparrow1$, the left-hand side converges to
$P_{\mathbf b\leftarrow\mathbf a}(1)$.  On the right, contractivity of
$\Phi_t$ permits
dominated convergence, and \cref{prop:transfer} identifies its
almost-everywhere boundary value as $e^{-\iu tH_\theta}$.  This proves
\cref{eq:boundary-phase-integral}.
\end{proof}

When $\mathbf b=\mathbf a$, the clock cross-spectrum becomes the clock
power spectrum
\[
 \rho_{\mathbf a\leftarrow\mathbf a}(\theta)
 =
 \frac{|a(e^{\iu\theta})|^2}{2\pi}\ge0,
\]
whose Fourier coefficients are the autocorrelation coefficients of
the clock state.

The auxiliary private Hilbert space is
$\QQ=L^2([0,t];\EE)$ from \cref{eq:history}.  A vector $q\in\QQ$ is
the private component of a vector $(\xi,q)\in\HH\oplus\QQ$; in the
continuous coordinate representation it is described by the
$\EE$-valued wavefunction $r\mapsto q(r)$, with
\[
 \norm{q}_{\QQ}^2=\int_0^t\norm{q(r)}_{\EE}^2\,dr.
\]
Thus $q(r)$ is an amplitude profile, not a classical history, and
$q$ need not be normalized separately.  The term history refers only
to the causal dependence in \cref{eq:psi-def}: the public state at
coordinate $r$ depends on field values $q(u)$ only for $u\le r$.
This private state space is distinct from the finite routing clock
$\bbc_C^R$, which labels the discrete uses of the shared private
space.  The full space $\QQ$ provides a single horizon-independent
analytic input--output realization.  The bounded-horizon circuit
accesses only finite-dimensional subspaces, using the reachable-history
restriction and unitary-completion device of
\cite[Appendix~A.3]{ChenGaoWangZhou2026}.  In the present
time-independent setting these subspaces take the following form.  For
$k\ge0$, define
\begin{equation}\label{eq:Fk}
 \mathcal F_k=\{r\mapsto e^{-\sqrt d\,\Lambda r}f(r):
 f\text{ is an $\EE$-valued polynomial of degree at most }k\},
 \qquad \mathcal F_{-1}=\{0\}.
\end{equation}
\crefrange{eq:B}{eq:A0} implies
\begin{equation}\label{eq:reachable}
 B\HH\subseteq\mathcal F_0,
 \qquad
 A_0\mathcal F_k\subseteq\mathcal F_{k+1},
\end{equation}
and the pointwise $U_M$ preserves every $\mathcal F_k$.

Although $\QQ$ is infinite dimensional and is not itself a finite
quantum register, a routed computation of bounded horizon reaches only
the finite-dimensional ladder
$\mathcal F_0\subseteq\cdots\subseteq\mathcal F_{R-1}$, so it admits
an exact finite-dimensional realization on the subspace that matters.
The following lemma supplies this realization and makes later uses of
both the circuit and its adjoint legitimate.

\begin{lemma}[Finite reachable-history completion]
\label{lem:finite-completion}
For a horizon of $R$ forward calls, the restriction
\[
 W_t:\HH\oplus\mathcal F_{R-2}\longrightarrow
      \HH\oplus\mathcal F_{R-1}
\]
extends to a finite-dimensional unitary
$W_{t,R}^{\mathrm{fin}}$ on $\HH\oplus\mathcal F_{R-1}$.  The extension
depends on $d,\Lambda,t,V,R$ but not on $M$, and agrees exactly with the
continuous construction on every state reached in the forward routed
circuit.
\end{lemma}

\begin{proof}
The displayed restriction is an isometry by \cref{lem:Wt}.  Complete
orthonormal bases of its domain and range complements.  The resulting
unitary completion is independent of $M$ because $W_t$ is.  Starting
from zero private amplitude,~\cref{eq:reachable} shows that the state
before the last call belongs to $\mathcal F_{R-2}$, so all routed
forward outputs agree exactly, and therefore so does the compressed
block $P_{\mathbf b\leftarrow\mathbf a}(z)$ from \cref{eq:Pa}.
Later singular-value transformations use this genuine finite unitary
and its genuine adjoint; no inverse continuous evolution on unreachable
histories is assumed.
\end{proof}

\section{Logarithmic precision from a phase-twisted clock}
\label{sec:logarithmic}

Throughout this section let $M=M^\dagger$ act on $\EE$, let
$V:\HH\to\EE$ be an isometry, and assume $t\ge0$, $\Lambda>0$, and
$d\ge1$.  Suppose
\begin{equation}\label{eq:abstract-assumptions}
 H=2\sqrt d\,\Lambda V^\dagger MV,
 \qquad
 V^\dagger M^2V\preceq\frac14\id,
 \qquad
 L_0\ge\max\{1,\norm{M}\},
 \qquad
 \tau=t\sqrt d\,\Lambda.
\end{equation}
Here $L_0$ is a known pole-control parameter for the resolvent
$(\id+2uM)^{-1}$; it is not a spectral-norm bound on $H$.  For the
sparse-entry lifting, \cref{eq:Mnorm} gives the known choice
\[
 L_0=\max\left\{1,
             \frac{\sqrt d\,\Lambda_{\max}}{2\Lambda}\right\}.
\]
Using the separate entry bound here is important: increasing the working
normalization decreases this pole-control parameter, whereas it
increases $\tau$.

\subsection{A localized reproducing phase kernel}

By \cref{eq:boundary-phase-integral}, a finite input--output clock pair
produces a boundary-phase integral whose kernel is the clock
cross-spectrum.  Since $H_0=H$, Hamiltonian simulation therefore
amounts to constructing finite clock states whose cross-spectrum acts
as an approximate evaluation kernel at $\theta=0$ for the boundary
family.  An exact delta function at zero would require infinitely many
clock slots, so we construct a signed phase kernel of unit integral
that is localized near zero and exactly reproduces sufficiently many
low Fourier modes.  The resulting signed phase integral is a linear
combination of the boundary evolutions.

Here localization means that the magnitude of the kernel is strongly
suppressed when $|\sin(\theta/2)|$ is larger than $O(1/m)$.
This suppression controls boundary responses far from the target
phase, where the local Taylor approximation is unavailable.  Exact
reproduction then cancels the low-order local expansion, leaving an
error that decreases exponentially with the reproduction order $p$.
The single signed phase kernel below has both properties and, as shown
later, factors as the cross-spectrum of the input and output clock
states.

The construction combines a classical de la Vall\'ee Poussin kernel,
formed from Fej\'er kernels
\cite[Eqs.~(1.13)--(1.15)]{Mehta2013}, with an additional high-power
localization factor.  We prove the resulting reproduction and
approximation bounds below.

For positive integers $m,p$, set $K=mp$ and define the shaping
polynomial
\begin{equation}\label{eq:shaping-polynomial}
 \begin{aligned}
 c(z)
 &=
 \left(\sum_{n=0}^{K-1}z^n\right)
 \left(\frac1m\sum_{n=0}^{m-1}z^n\right)^p
 =\sum_{j=0}^{L_c-1}c_jz^j,\\
 L_c&:=K+p(m-1)=2mp-p.
 \end{aligned}
\end{equation}
The coefficients $c_j$ are real and nonnegative.  Throughout, we use
the zero extension $c_j=0$ outside $0\le j<L_c$, and all unqualified
coefficient sums over $j$ are understood over $\mathbb Z$.
Define
\begin{equation}\label{eq:phase-kernel}
 \kappa_{m,p}(\theta)
 =
 \frac{|c(e^{\iu\theta})|^2}{2\pi K}
 \bigl(1+2\cos(K\theta)\bigr),
\end{equation}
and define the corresponding signed phase integral
\begin{equation}\label{eq:Ump}
 \widetilde U_{m,p}
 =
 \int_{-\pi}^{\pi}
 \kappa_{m,p}(\theta)e^{-\iu tH_\theta}\,d\theta.
\end{equation}
Here and below, the integrand denotes the almost-everywhere radial
boundary value from \cref{prop:transfer}; assigning arbitrary values at
the finitely many exceptional phases does not change the integral.

The normalized geometric sum of length $m$ supplies localization,
while the three frequency shifts enforce exact reproduction.  A local
even Taylor expansion then converts these properties into exponential
convergence.  The following proposition collects the three facts
needed by the clock construction.

\begin{proposition}[Localized phase-kernel approximation]
\label{prop:filtered-error}
The phase kernel $\kappa_{m,p}$ is real and even, and
\begin{align}
 \int_{-\pi}^{\pi}|\kappa_{m,p}(\theta)|\,d\theta
 &\le3,\label{eq:kappa-L1}\\
 \int_{-\pi}^{\pi}
 \kappa_{m,p}(\theta)e^{\iu k\theta}\,d\theta
 &=1,\qquad k\in\mathbb Z,\ |k|\le p.
 \label{eq:kappa-reproduce}
\end{align}
Under the assumptions in \cref{eq:abstract-assumptions},
\begin{equation}\label{eq:filtered-error}
 \norm{\widetilde U_{m,p}-e^{-\iu tH}}
 \le
 7e^{\tau/(2L_0)}
 \left(\frac{32L_0}{m}\right)^{2p}.
\end{equation}
In particular, $m=\lceil128L_0\rceil$ gives
\begin{equation}\label{eq:filtered-error-16}
 \norm{\widetilde U_{m,p}-e^{-\iu tH}}
 \le7e^{\tau/(2L_0)}16^{-p}.
\end{equation}
\end{proposition}

\begin{proof}
\emph{Reproduction and localization.}
Put
\[
 \beta(\theta)
 =
 \left|
 \frac1m\sum_{n=0}^{m-1}e^{\iu n\theta}
 \right|^{2p}
 =
 \sum_{|\ell|\le K-p}\widehat\beta_\ell e^{\iu\ell\theta}.
\]
Its coefficients satisfy
$\sum_\ell\widehat\beta_\ell=\beta(0)=1$.  Directly
expanding the geometric sum of length $K$ gives, for every integer
$n$,
\begin{equation}\label{eq:geometric-moment}
 \int_{-\pi}^{\pi}
 \frac{\left|\sum_{j=0}^{K-1}e^{\iu j\theta}\right|^2}
 {2\pi K}
 e^{\iu n\theta}\,d\theta
 =
 \phi_K(n),
 \qquad
 \phi_K(n):=\left(1-\frac{|n|}{K}\right)_+,
 \qquad
 (y)_+:=\max\{y,0\}.
\end{equation}
If $|n|\le K$, then
\[
 \phi_K(n)+\phi_K(n+K)+\phi_K(n-K)=1.
\]
For $|k|\le p$ and $|\ell|\le K-p$, one has
$|k+\ell|\le K$.  Expanding $\beta$ in
\cref{eq:phase-kernel}, applying \cref{eq:geometric-moment} to its
three shifted terms, and summing over $\ell$ therefore proves
\cref{eq:kappa-reproduce}.

The reproduction identity has the following useful interpretation.
If
\[
 f(\theta)=\sum_{|k|\le p}f_ke^{\iu k\theta}
\]
is any operator-valued trigonometric polynomial of degree at most
$p$, then
\[
 \begin{aligned}
 \int_{-\pi}^{\pi}
 \kappa_{m,p}(\theta)f(\theta)\,d\theta
 &=
 \sum_{|k|\le p}f_k
 \int_{-\pi}^{\pi}
 \kappa_{m,p}(\theta)e^{\iu k\theta}\,d\theta\\
 &=
 \sum_{|k|\le p}f_k
 =f(0).
 \end{aligned}
\]
Thus integration against $\kappa_{m,p}$ agrees exactly with
evaluation at $\theta=0$ on trigonometric polynomials of degree at
most $p$.

The same geometric-sum formula gives
\begin{equation}\label{eq:phase-localization}
 0\le\beta(\theta)
 \le
 \min\left\{
 1,\frac1{(m|\sin(\theta/2)|)^{2p}}
 \right\}.
\end{equation}
The density on the left-hand side of \cref{eq:geometric-moment} is
nonnegative and integrates to one.  Since
$|1+2\cos(K\theta)|\le3$ and $\beta\le1$, this proves
\cref{eq:kappa-L1}.  Reality and evenness are immediate from
\cref{eq:phase-kernel}.

\emph{Local analytic expansion.}
Set $r_\ast=(16L_0)^{-1}$.  For complex $x$ near zero, use the analytic
branches of
\[
 u(x)=\frac{x}{\sqrt{1-x^2}},
 \qquad
 \mathcal G(u):=(\id+2uM)^{-1}.
\]
The rational formula for $H_\theta$ in \cref{prop:transfer} defines
its complex continuation near $x=0$ and gives the exact identities
\begin{align*}
 (2M-u\id)\mathcal G(u)-2M
 &=-u(\id+4M^2)\mathcal G(u),\\
 t(H_\theta-H)
 &=-\tau uV^\dagger(\id+4M^2)\mathcal G(u)V,
 \qquad u=\tan(\theta/2).
\end{align*}
If $2|u|\norm{M}\le1/2$, then $\norm{\mathcal G(u)}\le2$.  Since
$\mathcal G(u)$
commutes with $M$ and $\norm{MV}\le1/2$,
\[
 \norm{V^\dagger(\id+4M^2)\mathcal G(u)V}
 \le\norm{\mathcal G(u)}(1+4\norm{MV}^2)\le4.
\]
For $|x|\le r_\ast\le1/16$, one has $|u(x)|\le2r_\ast$ and
$2|u(x)|\norm{M}\le1/4$.  The resolvents for both signs of $u(x)$
are therefore analytic in a neighborhood of the closed disk, and
\[
 \norm{t(H_{\pm2\arcsin x}-H)}
 \le8\tau r_\ast=\frac{\tau}{2L_0}.
\]
The Duhamel bound $\norm{e^{X+Y}}\le e^{\norm{Y}}$ for
skew-Hermitian $X$ now shows that the even analytic function
\[
 h(x)
 :=
 \frac12\left(
 e^{-\iu tH_{2\arcsin x}}
 +
 e^{-\iu tH_{-2\arcsin x}}
 \right)
\]
obeys $\sup_{|x|\le r_\ast}\norm{h(x)}\le e^{\tau/(2L_0)}$.
The operator-valued Cauchy estimate consequently gives
\begin{equation}\label{eq:local-even-series}
 h(x)=\sum_{j\ge0}h_jx^{2j},
 \qquad
 h_0=e^{-\iu tH},
 \qquad
 \norm{h_j}\le e^{\tau/(2L_0)}r_\ast^{-2j}.
\end{equation}

\emph{Approximation error.}
For almost every real $\theta\in[-\pi,\pi]$, set
$x=\sin(\theta/2)$.  On this interval $\arcsin x=\theta/2$, so
$h(x)$ is the even part of the boundary response.  Define its
truncation
\[
 h_{<p}(\theta)
 =
 \sum_{j=0}^{p-1}h_j\sin^{2j}(\theta/2).
\]
Two facts connect this Taylor truncation to the reproduction identity.
First, because
\[
 \sin^2(\theta/2)
 =
 \frac{2-e^{\iu\theta}-e^{-\iu\theta}}4,
\]
the function $h_{<p}$ is a trigonometric polynomial of degree at
most $p-1$.  Second,
\[
 h_{<p}(0)=h_0=e^{-\iu tH},
\]
since every term with $j\ge1$ contains a positive power of
$\sin(\theta/2)$.  The evaluation rule following
\cref{eq:kappa-reproduce} therefore gives
\[
 \int_{-\pi}^{\pi}
 \kappa_{m,p}(\theta)h_{<p}(\theta)\,d\theta
 =
 h_{<p}(0)
 =e^{-\iu tH}.
\]

Since $\kappa_{m,p}$ is even, the boundary response in
\cref{eq:Ump} may be replaced by its even part $h(x)$, where
$x=\sin(\theta/2)$.  Combining this observation with the preceding
exact reproduction gives the central error identity
\begin{equation}\label{eq:central-error-identity}
 \widetilde U_{m,p}-e^{-\iu tH}
 =
 \int_{-\pi}^{\pi}
 \kappa_{m,p}(\theta)
 \bigl(h(x)-h_{<p}(\theta)\bigr)\,d\theta.
\end{equation}
Thus it remains to control the Taylor remainder after localization:
the local Taylor estimate handles $|x|\le r_\ast/2$, while
$\beta(\theta)$ suppresses the complementary region.

When $|x|\le r_\ast/2$, \cref{eq:local-even-series} yields
\[
 \norm{h(x)-h_{<p}(\theta)}
 \le
 \frac43e^{\tau/(2L_0)}
 \left(\frac{|x|}{r_\ast}\right)^{2p}.
\]
Multiplication by \cref{eq:phase-localization} gives, with the value
at $x=0$ understood by continuity,
\[
 \beta(\theta)\norm{h(x)-h_{<p}(\theta)}
 \le
 \frac43e^{\tau/(2L_0)}(mr_\ast)^{-2p}.
\]
When $|x|\ge r_\ast/2$, the boundary response from
\cref{prop:transfer} is unitary for almost every $\theta$, so
$\norm{h(x)}\le1$.  Moreover,
\begin{align*}
 \beta(\theta)\norm{h_{<p}(\theta)}
 &\le
 e^{\tau/(2L_0)}(mr_\ast)^{-2p}
 \sum_{j=0}^{p-1}
 \left(\frac {r_\ast}{|x|}\right)^{2(p-j)}\\
 &\le
 \frac43e^{\tau/(2L_0)}
 \left(\frac2{mr_\ast}\right)^{2p},
\end{align*}
while $\beta(\theta)\norm{h(x)}\le(2/(mr_\ast))^{2p}$.
Because $e^{\tau/(2L_0)}\ge1$, both regions therefore satisfy
\[
 \beta(\theta)\norm{h(x)-h_{<p}(\theta)}
 \le
 \frac73e^{\tau/(2L_0)}
 \left(\frac2{mr_\ast}\right)^{2p}.
\]
Using \cref{eq:central-error-identity,eq:phase-kernel} and the
factorization
\[
 |c(e^{\iu\theta})|^2
 =
 \left|\sum_{j=0}^{K-1}e^{\iu j\theta}\right|^2
 \beta(\theta),
\]
we obtain
\begin{align*}
 \norm{\widetilde U_{m,p}-e^{-\iu tH}}
 &\le
 \int_{-\pi}^{\pi}
 \frac{\left|\sum_{j=0}^{K-1}e^{\iu j\theta}\right|^2}
      {2\pi K}
 |1+2\cos(K\theta)|
 \beta(\theta)
 \norm{h(x)-h_{<p}(\theta)}
 \,d\theta\\
 &\le
 3\sup_{\theta\in[-\pi,\pi]}
 \left[
 \beta(\theta)
 \norm{h(x)-h_{<p}(\theta)}
 \right]\\
 &\le
 7e^{\tau/(2L_0)}
 \left(\frac2{mr_\ast}\right)^{2p}.
\end{align*}
Here the second inequality uses
$|1+2\cos(K\theta)|\le3$ and the unit integral in
\cref{eq:geometric-moment}.  Since $r_\ast=(16L_0)^{-1}$, one has
$2/(mr_\ast)=32L_0/m$, proving \cref{eq:filtered-error}.
If $m=\lceil128L_0\rceil$, then $mr_\ast\ge8$, so
$(2/(mr_\ast))^{2p}\le16^{-p}$, proving
\cref{eq:filtered-error-16}.
\end{proof}

\subsection{Single-clock realization of the signed phase kernel}

A clock circuit using the same input and output state has a
nonnegative clock power spectrum and therefore cannot realize the
signed factor in \cref{eq:phase-kernel}.  The key observation is that
one routed clock circuit may use different input and output clock
states.  A phase twist then factors the entire signed phase kernel as
one clock cross-spectrum, avoiding an LCU over two routed circuits.
We call this input--output pair, implemented on a single routing clock
register, a phase-twisted clock.

Let $\omega_{\rm tw}=e^{\iu\pi/3}$ and put
\begin{equation}\label{eq:lambda}
 \lambda
 :=
 \frac1K\left(
 2\sum_jc_j^2+\sum_jc_{j+K}c_j
 \right).
\end{equation}
Define the coefficient-normalized generating polynomials of the input
and output clock states by
\begin{equation}\label{eq:clock-polynomials}
 a(z)=\frac{c(z)(\omega_{\rm tw}+z^K)}{\sqrt{K\lambda}},
 \qquad
 b(z)=\frac{c(z)(1+\omega_{\rm tw}z^K)}{\sqrt{K\lambda}}.
\end{equation}
Equivalently, their coefficients are
\begin{equation}\label{eq:clock-coefficients}
 a_j=\frac{\omega_{\rm tw}c_j+c_{j-K}}{\sqrt{K\lambda}},
 \qquad
 b_j=\frac{c_j+\omega_{\rm tw}c_{j-K}}{\sqrt{K\lambda}},
 \qquad
 0\le j<R.
\end{equation}
Let $\mathbf a=(a_0,\ldots,a_{R-1})$ and
$\mathbf b=(b_0,\ldots,b_{R-1})$ be their coefficient vectors, padded
to
the common horizon
\begin{equation}\label{eq:horizons}
 R=L_c+K=3mp-p.
\end{equation}

The normalization in \cref{eq:lambda} is chosen so that both
clock-state amplitude vectors have unit norm, while the opposite
placements of $\omega_{\rm tw}$ make their clock cross-spectrum equal the
signed phase kernel
$\kappa_{m,p}/\lambda$.  The following lemma verifies these facts and
identifies the resulting clock block.

\begin{lemma}[Phase-twisted input--output clock realization]
\label{lem:single-clock}
The vectors $\mathbf a=(a_j)_{j=0}^{R-1}$ and
$\mathbf b=(b_j)_{j=0}^{R-1}$ defined by
\cref{eq:clock-coefficients}, with the common horizon $R$ from
\cref{eq:horizons}, are normalized, and $1\le\lambda\le2$.  In the
notation of \cref{eq:Pa}, their
prepare--route--unprepare block satisfies
\begin{equation}\label{eq:single-clock-block}
 P_{\mathbf b\leftarrow\mathbf a}(1)
 =J_{\mathbf b}^\dagger\mathcal R_R(1)J_{\mathbf a}
 =\frac{\widetilde U_{m,p}}{\lambda}.
\end{equation}
\end{lemma}

\begin{proof}
For brevity, set
\[
 \Sigma_0=\sum_jc_j^2,
 \qquad
 \Sigma_K=\sum_jc_{j+K}c_j.
\]
The zero-frequency case of \cref{eq:kappa-reproduce}, together with
\cref{eq:phase-kernel} and Parseval's identity, gives
\begin{equation}\label{eq:clock-overlap-identity}
 \Sigma_0+2\Sigma_K=K.
\end{equation}
The coefficients $c_j$ are nonnegative, are supported on $L_c$
positions, and sum to $c(1)=K$.  Hence Cauchy--Schwarz gives
\[
\frac{\Sigma_0}{K}
 \ge\frac K{L_c}
 =\frac{m}{2m-1}>\frac12.
\]
Since the coefficients are nonnegative, $\Sigma_K\ge0$, so
\cref{eq:clock-overlap-identity} also gives $\Sigma_0\le K$.
By \cref{eq:clock-overlap-identity},
\[
 \lambda
 =\frac{2\Sigma_0+\Sigma_K}{K}
 =\frac12+\frac{3\Sigma_0}{2K},
\]
which proves $1\le\lambda\le2$.

Here $\norm{\cdot}_{\ell^2}$ denotes the Euclidean norm of a
polynomial's coefficient vector.  Since
$\omega_{\rm tw}+\overline{\omega_{\rm tw}}=1$, direct expansion gives
\[
 \norm{c(z)(\omega_{\rm tw}+z^K)}_{\ell^2}^2
 =
 \norm{c(z)(1+\omega_{\rm tw}z^K)}_{\ell^2}^2
 =
 2\Sigma_0+\Sigma_K=K\lambda.
\]
Thus $\norm{\mathbf a}_2=\norm{\mathbf b}_2=1$.
For $z=e^{\iu\theta}$,
\begin{align}
 \overline{b(z)}a(z)
 &=\frac{|c(z)|^2}{K\lambda}
   (1+\overline{\omega_{\rm tw}}z^{-K})
   (\omega_{\rm tw}+z^K)\nonumber\\
 &=\frac{|c(z)|^2}{K\lambda}(1+z^K+z^{-K})
 =\frac{2\pi}{\lambda}\kappa_{m,p}(\theta).
 \label{eq:cross-factorization}
\end{align}
Hence the clock cross-spectrum
$\rho_{\mathbf b\leftarrow\mathbf a}$ in \cref{eq:rhoa} is exactly
$\kappa_{m,p}/\lambda$.  Therefore
\cref{eq:boundary-phase-integral} gives
\cref{eq:single-clock-block}.
\end{proof}

By \cref{eq:clock-polynomials,eq:horizons}, each clock-state generating
polynomial has degree at most $R-1$.  One application of this block
encoding therefore uses exactly $R=3mp-p$ Cayley queries.

The proposition supplies an exponentially accurate signed phase
integral, and the lemma realizes it as a block encoding with bounded
normalization.  Oblivious amplitude amplification now converts
that block into a simulator.

\begin{theorem}[Abstract logarithmic-precision theorem]
\label{thm:abstract-log}
Let $M=M^\dagger$ act on a finite-dimensional space $\EE$, let
$V:\HH\to\EE$ satisfy $V^\dagger V=\id$, and assume $t\ge0$,
$\Lambda>0$, $d\ge1$, and \cref{eq:abstract-assumptions}.  Simulation
of $e^{-\iu tH}$ with operator-norm error $0<\eps\le1/2$ uses
\begin{equation}\label{eq:abstract-log-bound}
 O\!\left(\tau+L_0\log\frac{2}{\eps}\right)
\end{equation}
Cayley queries.
\end{theorem}

\begin{proof}
Choose
\begin{equation}\label{eq:mp-choice}
 m=\lceil128L_0\rceil,
 \qquad
 p=\left\lceil
 \frac{\tau/(2L_0)+\log(56/\eps)}{\log16}
 \right\rceil.
\end{equation}
Then \cref{eq:filtered-error-16} is at most $\eps/8$.  By
\cref{lem:single-clock}, one phase-twisted clock circuit has clean
block $\widetilde U_{m,p}/\lambda$, where
$1\le\lambda\le2$.  Apply \cref{lem:purification} to the approximation
$\widetilde U_{m,p}$ of $e^{-\iu tH}$.  One OAA
round converts this block to a simulator with operator-norm error
below $\eps/2$, leaving the other half of the budget for circuit
implementation.  The three calls to the length-$R$ routed clock circuit use at
most $3R=9mp-3p<9mp$ Cayley queries.  Since
\[
 m=O(L_0),
 \qquad
 p=O\!\left(1+\frac{\tau}{L_0}
                 +\log\frac{2}{\eps}\right),
\]
this proves \cref{eq:abstract-log-bound}.
\end{proof}

\begin{proof}[Proof of \cref{thm:main}]
For $t=0$, the identity circuit simulates the evolution exactly.
Suppose $t>0$.  Increasing the working normalization increases the
evolution scale but decreases the norm of the lifted operator.  The
latter enlarges the analytic neighborhood used by the approximation
argument and reduces the precision-dependent cost.  To balance these
contributions, set
\begin{equation}\label{eq:working-normalization}
 \Lambda'
 :=\max\left\{
 \Lambda,\,
 \sqrt{\frac{\Lambda_{\max}\log(2/\eps)}{t}},\,
 \frac{1}{t\sqrt d}
 \right\},
 \qquad
 \tau':=t\sqrt d\,\Lambda'.
\end{equation}
Apply the sparse-entry construction with $\Lambda'$ in place of
$\Lambda$.  The compressed moment remains bounded because
\[
 V^\dagger M^2V
 \preceq\frac{\Lambda^2}{4(\Lambda')^2}\id
 \preceq\frac14\id,
 \qquad
 L_0':=\max\left\{1,
        \frac{\sqrt d\,\Lambda_{\max}}{2\Lambda'}\right\}
 \ge\max\{1,\norm{M}\}.
\]
Each Cayley query costs $O(1)$ original sparse-oracle calls.  Therefore
\cref{thm:abstract-log}, applied with this working normalization and
$L_0'$, gives
\begin{equation}\label{eq:normalization-tradeoff}
 O\!\left(\tau'+L_0'\log\frac{2}{\eps}\right)
 =O\!\left(
 t\sqrt d\,\Lambda'+\log\frac{2}{\eps}
 +\frac{\sqrt d\,\Lambda_{\max}\log(2/\eps)}{\Lambda'}
 \right)
\end{equation}
sparse-oracle queries.  The choice in \cref{eq:working-normalization}
satisfies
\begin{align}
 t\sqrt d\,\Lambda'
 &\le t\sqrt d\,\Lambda
       +\sqrt{td\Lambda_{\max}\log\frac{2}{\eps}}+1,
 \label{eq:optimized-evolution-scale}\\
 \frac{\sqrt d\,\Lambda_{\max}\log(2/\eps)}{\Lambda'}
 &\le\sqrt{td\Lambda_{\max}\log\frac{2}{\eps}}.
 \label{eq:optimized-precision-scale}
\end{align}
Since $\log(2/\eps)\ge1$, these estimates prove
\cref{eq:query-bound}.  The last term in
\cref{eq:working-normalization} ensures $\tau'\ge1$ and adds only a
constant to the query bound.  A constant-factor upward rounding of
$\Lambda'$ preserves the compressed-moment bound and changes the
complexity by at most a constant factor.

Finally, if $t\Lambda\ge\log(2/\eps)$, then
$\Lambda_{\max}\le\Lambda$ implies
\[
 \sqrt{td\Lambda_{\max}\log(2/\eps)}
 \le t\sqrt d\,\Lambda,
 \qquad
 \log(2/\eps)\le t\sqrt d\,\Lambda.
\]
Thus the query complexity is $O(t\sqrt d\,\Lambda)$ in this regime.
\end{proof}

\section{Finite-gate realization of the coupling}
\label{sec:gate-realization}

The arbitrary unitary completion in \cref{lem:finite-completion} is
sufficient for a query upper bound but need not admit an efficient gate
decomposition.  We now replace it by a structured Fourier-diagonal
unitary.  The approximation is required only on the finite family of
histories that can occur in a forward clock computation.
We present the circuit construction and resource estimates for a
general working normalization $\Lambda$.  To prove \cref{thm:gate},
we set $\Lambda=\Lambda'$, with $\Lambda'$ defined in
\cref{eq:working-normalization}.  We treat $t>0$ and assume
$\tau\ge1$ for the resource accounting, and the chosen normalization
ensures this condition.  At $t=0$, the identity circuit suffices.

\subsection{An explicit finite coupling circuit}

The address workspace decomposes orthogonally as
\begin{equation}\label{eq:active-decomposition}
 \EE=V\HH\oplus(V\HH)^\perp.
\end{equation}
\crefrange{eq:C0}{eq:A0} show that $W_t$ is the
identity on pointwise fields in $(V\HH)^\perp$.  On $V\HH$, use
$V\xi\leftrightarrow\xi$, set
\begin{equation}\label{eq:scaled-coordinate}
 \begin{aligned}
 x&=\sqrt d\,\Lambda r,
 &0\le x\le\tau,\\
 \widetilde\psi(x)&=
 \psi\!\left(\frac{x}{\sqrt d\,\Lambda}\right),\\
 \widetilde q(x)&=(\sqrt d\,\Lambda)^{-1/2}
 q\!\left(\frac{x}{\sqrt d\,\Lambda}\right),
 &\widetilde g(x)&=(\sqrt d\,\Lambda)^{-1/2}
 g\!\left(\frac{x}{\sqrt d\,\Lambda}\right).
 \end{aligned}
\end{equation}
and suppress both the unchanged system label and the tildes below.
The active part of the coupling is the scalar-channel unitary $w_\tau$:
\begin{align}
 \psi(\tau)
 &=e^{-\tau}\xi-\iu\sqrt2\int_0^\tau
 e^{-(\tau-x)}q(x)\,dx,\label{eq:wtau-public}\\
 g(x)
 &=q(x)-\iu\sqrt2e^{-x}\xi
 -2\int_0^xe^{-(x-y)}q(y)\,dy.\label{eq:wtau-history}
\end{align}
Thus \cref{eq:wtau-public,eq:wtau-history} are precisely
\cref{eq:psi-def,eq:Wt-def} restricted to the active subspace
$V\HH$ and conjugated by the isometric rescaling
\cref{eq:scaled-coordinate}; the component in $(V\HH)^\perp$
remains unchanged.
Every system amplitude undergoes the same scalar map.  We now give the
circuit that approximates this map.  The continuous all-pass
representation that motivates its Fourier phase and identifies its
target is deferred to the proof of \cref{lem:coupling-comparison}.

Choose an even integer $G\ge\max\{2,2\tau\}$, put
$\Delta=\tau/G\le1/2$, and set
\begin{equation}\label{eq:abdelta}
 r_\Delta=e^{-\Delta},
 \qquad \sigma_\Delta=\sqrt{1-r_\Delta^2}.
\end{equation}
Define the unit-modulus phase
\begin{equation}\label{eq:sa}
 s_\Delta(e^{\iu\theta})
 =\frac{r_\Delta-e^{-\iu\theta}}
 {1-r_\Delta e^{-\iu\theta}}.
\end{equation}
Choose an integer $L_{\rm io}\ge1$ such that
$r_\Delta^{L_{\rm io}}\le1/2$ and
$M_{\rm cyc}:=G+2L_{\rm io}$ is a power of two, and partition the cyclic
workspace into an $L_{\rm io}$-cell input interval, a $G$-cell
interaction interval, and an $L_{\rm io}$-cell output interval.  Its
basis is indexed modulo $M_{\rm cyc}$ as
\[
 -L_{\rm io},\ldots,-1,\quad 0,\ldots,G-1,\quad
 G,\ldots,G+L_{\rm io}-1.
\]
For any positive integer $M$, define the quantum Fourier transform by
\[
 \mathsf F_M\ket j
 :=\frac1{\sqrt M}\sum_{k=0}^{M-1}
 e^{-2\pi\iu jk/M}\ket k.
\]
When $M$ is a power of two, the standard QFT circuit uses
$O((\log M)^2)$ 1- and 2-qubit gates.
Define
\begin{align}
 \mathsf D_{\Delta,M_{\rm cyc}}
 &:=\diag\!\left(
 s_\Delta(e^{2\pi\iu k/M_{\rm cyc}})
 \right)_{k=0}^{M_{\rm cyc}-1},\label{eq:cyclic-phase}\\
 \mathsf A^{\rm cyc}_{\Delta,M_{\rm cyc}}
 &:=\mathsf F_{M_{\rm cyc}}^\dagger
 \mathsf D_{\Delta,M_{\rm cyc}}\mathsf F_{M_{\rm cyc}}.
 \label{eq:cyclic-allpass}
\end{align}
This is an exact finite unitary before finite-precision implementation.

Fix the distinguished address $\mathsf p=-1$ modulo $M_{\rm cyc}$
to encode the logical public mode, and define the
logical inclusion
\begin{equation}\label{eq:grid-inclusion}
 \iota_{\rm grid}(\xi,q_0,\ldots,q_{G-1})
 =\ket{\mathsf p}\xi+
  \sum_{j=0}^{G-1}\ket j q_j.
\end{equation}
The normalized geometric input and output encoding states are
\begin{align}
 \ket{\chi_{\rm in}}
 &=\frac{\iu\sigma_\Delta}
 {\sqrt{1-r_\Delta^{2L_{\rm io}}}}
   \sum_{h=0}^{L_{\rm io}-1}r_\Delta^h\ket{-h-1},
 \label{eq:finite-input-encoding-state}\\
 \ket{\chi_{\rm out}}
 &=-\frac{\iu\sigma_\Delta}
 {\sqrt{1-r_\Delta^{2L_{\rm io}}}}
   \sum_{h=0}^{L_{\rm io}-1}r_\Delta^h\ket{G+h}.
 \label{eq:finite-output-encoding-state}
\end{align}
We need efficient preparations that fix every interaction mode, not
merely preparations with the prescribed action on one input state.
To construct them explicitly, let
\[
 \ket{v_{\rm in}}=-\iu\ket{\chi_{\rm in}},
 \qquad
 \ket{v_{\rm out}}=\iu\ket{\chi_{\rm out}},
\]
so both vectors have real nonnegative geometric coefficients.  For
$v\in\{v_{\rm in},v_{\rm out}\}$, write
$v_{\mathsf p}=\langle\mathsf p|v\rangle$ and define
\begin{equation}\label{eq:encoding-householder}
 \ket{\phi_v}
 :=\frac{\ket{\mathsf p}-\ket v}{\sqrt{2(1-v_{\mathsf p})}},
 \qquad
 \mathsf H_v:=\id-2\ket{\phi_v}\!\bra{\phi_v}.
\end{equation}
These formulas are nonsingular with an absolute conditioning bound.
Indeed, $\Delta\le1/2$ and $r_\Delta^{L_{\rm io}}\le1/2$ imply
$L_{\rm io}\ge2$, and
\[
 (v_{\rm in})_{\mathsf p}
 =\left(\sum_{h=0}^{L_{\rm io}-1}r_\Delta^{2h}\right)^{-1/2}
 \le(1+e^{-1})^{-1/2}<1,
 \qquad
 (v_{\rm out})_{\mathsf p}=0.
\]
Let $\Pi_{\rm io}$ project onto the union of the input and output
intervals, and put
\begin{equation}\label{eq:encoding-completions}
 \begin{aligned}
  \mathsf D_\pm&:=(\id-\Pi_{\rm io})\pm\iu\Pi_{\rm io},\\
  U_{\rm in}&:=\mathsf D_+\mathsf H_{v_{\rm in}},
  \qquad
  U_{\rm out}:=\mathsf D_-\mathsf H_{v_{\rm out}}.
 \end{aligned}
\end{equation}
Since $\mathsf H_v\ket{\mathsf p}=\ket v$ and
$\phi_v$ is supported outside the interaction interval, these full
unitaries satisfy
\[
 U_{\rm in}\ket{\mathsf p}=\ket{\chi_{\rm in}},
 \qquad
 U_{\rm out}\ket{\mathsf p}=\ket{\chi_{\rm out}},
 \qquad
 U_{\rm in}\ket j=U_{\rm out}\ket j=\ket j
 \quad(0\le j<G).
\]
Their actions on the remaining input and output modes are now specified
by \cref{eq:encoding-completions}.  These completion modes are part
of the finite workspace and need not be fixed pointwise.  The
implemented scalar coupling is the full unitary
\begin{equation}\label{eq:finite-coupling}
 \widetilde w_{\Delta,\tau,L_{\rm io}}
 :=U_{\rm out}^\dagger\mathsf A^{\rm cyc}_{\Delta,M_{\rm cyc}}
 U_{\rm in}.
\end{equation}
Thus the circuit prepares the geometric input encoding state, applies
a QFT, computes and applies the diagonal phase, applies the inverse
QFT, and unprepares the geometric output encoding state.

The two encoding states can be prepared without a length-$L_{\rm io}$
table.  For every binary interval
$[u,v]\subseteq\{0,\ldots,L_{\rm io}-1\}$, its
unnormalized probability mass is
\begin{equation}\label{eq:geometric-interval-mass}
 \sum_{h=u}^vr_\Delta^{2h}
 =r_\Delta^{2u}
 \frac{1-r_\Delta^{2(v-u+1)}}{1-r_\Delta^2}.
\end{equation}
The same closed form prepares the Householder vectors.  For every
binary address interval $I$, their probability mass is
\begin{equation}\label{eq:householder-interval-mass}
 \sum_{j\in I}|(\phi_v)_j|^2
 =\frac{
  \sum_{\substack{j\in I\\j\ne\mathsf p}}v_j^2
  +\mathbf1_{\{\mathsf p\in I\}}(1-v_{\mathsf p})^2
 }{2(1-v_{\mathsf p})}.
\end{equation}
The sum in the numerator is obtained from at most two geometric
interval masses after intersecting $I$ with the appropriate encoding
interval and excluding $\mathsf p$.  Thus a binary
interval-mass recursion computes the two child masses, rotates one
address qubit, and uncomputes at each of $O(\log M_{\rm cyc})$ levels.
This is the interval-mass state-preparation method of Grover and
Rudolph \cite{GroverRudolph2002}, specialized to the explicit masses
above.
The signs of $\phi_v$ are supplied by an interval-membership phase and
a comparison with $\mathsf p$.  Intervals containing $\mathsf p$ have
a mass bounded below by an absolute constant.  For the other nonzero
intervals, analytically cancel the common geometric factor before
computing each conditional probability.  Exponentially small
contributions may be truncated as long as the resulting state-preparation
error stays within the allowed tolerance.  These
choices avoid computing exponentially small parent masses at relative
precision, so the recursion needs only the working precision charged
below.  If $B_v\ket0=\ket{\phi_v}$ is this
preparation, then
\[
 \mathsf H_v=B_v(\id-2\ket0\!\bra0)B_v^\dagger
\]
implements its full unitary completion.  Preparing $\phi_v$ to state
error $\delta_{\rm impl}$ changes this reflection by
$O(\delta_{\rm impl})$ in operator norm, independently of the
preparation circuit's action on other input states.  Comparisons
implement $\mathsf D_\pm$.  The diagonal phase in \cref{eq:cyclic-phase} is computed
from the frequency address and applied by phase kickback.
Since its denominator has modulus at least
$1-r_\Delta=\Theta(\Delta)$, one call implemented with operator-norm error
$\delta_{\rm impl}$
costs
\begin{equation}\label{eq:one-coupling-gates}
 \poly\!\left(
 \log M_{\rm cyc}+\log(1/\Delta)+\log(1/\delta_{\rm impl})
 \right)
\end{equation}
gates and uses a logarithmic-size grid address.

For the $\EE$-valued coupling, a uniform-list-index preparation maps
$d^{-1/2}\sum_{\ell=1}^d\ket\ell$ to a distinguished active flag.
Conditioned on that flag, the circuit applies
\cref{eq:finite-coupling}; it then undoes the list-index preparation.  The
inactive subspace $(V\HH)^\perp$ is left invariant.  For arbitrary
$d$, the uniform state and its inverse are implemented by reversible
comparison and controlled rotations to the same working precision as
the other arithmetic primitives.  We use the same symbols $w_\tau$ and
$\widetilde w_{\Delta,\tau,L_{\rm io}}$ for their $\EE$-valued
extensions: on the active sector $V\HH\cong\HH$ they act as the scalar
channel tensored with $\id_\HH$, while on $(V\HH)^\perp$ they act as
the identity.  More precisely, its finite local logical space is
\[
 \HH\oplus
 \left(\Span\{\ket j:
 j\in\mathbb Z_{M_{\rm cyc}}\setminus\{\mathsf p\}\}\otimes\EE\right).
\]
It is embedded in the cyclic address space tensored with $\EE$ by
mapping the public vector $\xi$ to $\ket{\mathsf p}V\xi$ and
each private vector $q_j$ to $\ket j q_j$.  The omitted subspace
$\ket{\mathsf p}\otimes(V\HH)^\perp$ may be included as an unused
identity sector.  This makes the active-channel construction a full
unitary without identifying the public space $\HH$ with all of $\EE$.
We likewise use $\iota_{\rm grid}$ for the induced logical inclusion
on system-space and address-workspace amplitudes; its public summand
is understood to contain this $V$ embedding.

\subsection{Accuracy on the forward reachable space}

For the phase-twisted input--output clock pair from
\cref{lem:single-clock}, the horizon $R=3mp-p$ from
\cref{eq:horizons} is the number of interactions with the
shared private space, and one-round OAA adds only a factor of
three.  By \cref{eq:mp-choice},
\begin{equation}\label{eq:R-query-bound}
 R=O(mp)
 =O\!\left(
 \tau+L_0\log\frac{2}{\eps}
 \right).
\end{equation}

For $k\ge0$, define the scalar reachable class
\begin{equation}\label{eq:scaled-reachable}
 \mathcal F^{\rm sc}_k(\tau)
 :=\{x\mapsto e^{-x}f(x):\deg f\le k\},
 \qquad
 \mathcal F^{\rm sc}_{-1}(\tau):=\{0\}.
\end{equation}
In scaled coordinates, every $\EE$-valued private field reachable before
step $R$ belongs to
$\mathcal F^{\rm sc}_{R-1}(\tau)\otimes\EE$.

A finite grid cannot approximate arbitrary $L^2$ histories uniformly,
so the analysis exploits the smaller class generated by forward
reuse.  Define the derivative-control factor
\begin{equation}\label{eq:reachable-K}
 L_R(\tau)
 :=1+\frac{R^2}{\min\{\tau,1\}}.
\end{equation}
The next lemma is the only continuous--finite comparison needed in
the main argument.  Its proof contains the full-line representation,
the sequential reference discretization, the truncation and
wrap-around estimate, and the construction of a comparison isometry.
None of those auxiliary
objects is implemented by the algorithm.
All implicit constants in $\lesssim$ below are absolute.

\begin{lemma}[Finite coupling accuracy on reachable histories]
\label{lem:coupling-comparison}
There is an absolute constant $c_0>0$ such that the following holds.
Suppose $\Delta L_R(\tau)\le c_0$ and
$r_\Delta^{L_{\rm io}}\le1/2$.  There is an isometry
\[
 \mathsf E_\Delta:
 \mathcal F^{\rm sc}_{R-1}(\tau)\longrightarrow\bbc^G
\]
acting only on the history coordinate.  Its $\EE$-valued extension
\[
 \mathsf E_\Delta^{(\EE)}
 :=\mathsf E_\Delta\otimes\id_\EE
\]
intertwines every pointwise operator $T$ on $\EE$:
\begin{equation}\label{eq:grid-intertwining}
 (\id_{\bbc^G}\otimes T)\mathsf E_\Delta^{(\EE)}
 =\mathsf E_\Delta^{(\EE)}T,
\end{equation}
where the $T$ on the right is the pointwise lift defined after
\cref{eq:history}.  On the public space plus
$\mathcal F^{\rm sc}_{R-2}(\tau)\otimes\EE$, where
$\iota_{\rm grid}$ is the
logical inclusion defined in~\cref{eq:grid-inclusion} and $w_\tau$ is the
$\EE$-valued extension of the scalar-channel unitary defined by
\cref{eq:wtau-public,eq:wtau-history},
\begin{equation}\label{eq:coupling-comparison}
 \norm{
 \widetilde w_{\Delta,\tau,L_{\rm io}}\iota_{\rm grid}
 (\id\oplus\mathsf E_\Delta^{(\EE)})
 -\iota_{\rm grid}(\id\oplus\mathsf E_\Delta^{(\EE)})w_\tau}
 \lesssim
 \Delta L_R(\tau)+\Delta^2
 +e^{-\Delta L_{\rm io}}
 .
\end{equation}
\end{lemma}

\begin{proof}
\emph{Step 1: full-line representation of the ideal coupling.}
On $L^2(\bbr)$, define the causal all-pass operator
\begin{equation}\label{eq:full-line-allpass}
 (\mathsf A_\infty f)(x)=f(x)-2\int_{-\infty}^x
 e^{-(x-y)}f(y)\,dy.
\end{equation}
For the Fourier convention
$\widehat f(\omega)=\int_\bbr f(x)e^{-\iu\omega x}\,dx$,
\begin{equation}\label{eq:allpass-multiplier}
 \widehat{\mathsf A_\infty f}(\omega)
 =\frac{\iu\omega-1}{\iu\omega+1}\widehat f(\omega).
\end{equation}
The multiplier has modulus one, so $\mathsf A_\infty$ is unitary by
Plancherel's theorem.  Define isometries into the full line by
\begin{align}
 (J_-(\xi,q))(x)
 &=\begin{cases}
   \iu\sqrt2e^x\xi,&x<0,\\
   q(x),&0\le x\le\tau,\\
   0,&x>\tau,
  \end{cases}\label{eq:Jminus}\\
 (J_+(\xi_{\rm out},g))(x)
 &=\begin{cases}
   0,&x<0,\\
   g(x),&0\le x\le\tau,\\
   -\iu\sqrt2e^{-(x-\tau)}\xi_{\rm out},&x>\tau.
  \end{cases}\label{eq:Jplus}
\end{align}
Both exponential functions have squared norm one.  Direct substitution
into \cref{eq:full-line-allpass} gives
\begin{equation}\label{eq:allpass-intertwining}
 \mathsf A_\infty J_-=J_+w_\tau.
\end{equation}
The negative-coordinate function cancels the output on $(-\infty,0)$
and supplies the public-input term on $[0,\tau]$; beyond $\tau$, the
convolution factors into the positive-coordinate function with
coefficient $\psi(\tau)$.  This identifies the
continuous all-pass target whose finite analogue is implemented in
\cref{eq:finite-coupling}.

\emph{Step 2: a reference cell discretization.}
The recurrence
\begin{equation}\label{eq:cell-recurrence}
 \begin{pmatrix}\zeta_{j+1}\\g_j\end{pmatrix}
 =\begin{pmatrix}
 r_\Delta&-\iu\sigma_\Delta\\
 -\iu\sigma_\Delta&r_\Delta
 \end{pmatrix}
  \begin{pmatrix}\zeta_j\\q_j\end{pmatrix},
 \qquad \zeta_0=\xi,
\end{equation}
defines the unitary reference map
\[
 w_{\Delta,G}:(\xi,q_0,\ldots,q_{G-1})\longmapsto
 (\zeta_G,g_0,\ldots,g_{G-1}).
\]
This map is used only in the proof; its $G$ cells are not executed
sequentially.  Eliminating the internal state on a bilateral sequence
gives
\begin{equation}\label{eq:discrete-allpass}
 (\mathsf A_\Delta q)_j
 =r_\Delta q_j-(1-r_\Delta^2)
 \sum_{k<j}r_\Delta^{j-k-1}q_k.
\end{equation}
Its Fourier multiplier is $s_\Delta$ from \cref{eq:sa}.  Because the
impulse response is absolutely summable, periodization commutes with
the filter, and the periodized action is exactly
$\mathsf A^{\rm cyc}_{\Delta,M_{\rm cyc}}$ from
\cref{eq:cyclic-allpass}.

\emph{Step 3: truncation and wrap-around comparison.}
The bilateral recurrence has the exact semi-infinite input and output
sequences
\begin{equation}
 q_{-h-1}=\iu\sigma_\Delta r_\Delta^h\xi,
 \qquad
 g_{G+h}=-\iu\sigma_\Delta r_\Delta^h\zeta_G,
 \qquad h\ge0.
\end{equation}
They are isometric because
$\sigma_\Delta^2\sum_{h\ge0}r_\Delta^{2h}=1$.  Periodizing with period
$M_{\rm cyc}$ folds a geometric remainder of norm at most
\begin{equation}
 \sum_{j\ge0}r_\Delta^{L_{\rm io}+jM_{\rm cyc}}\norm{\xi}
 =\frac{r_\Delta^{L_{\rm io}}}
 {1-r_\Delta^{M_{\rm cyc}}}\norm{\xi}
 \le\frac43r_\Delta^{L_{\rm io}}\norm{\xi},
\end{equation}
and the same estimate holds for the output sequence because
$\norm{\zeta_G}\le\norm{(\xi,q)}$.  Truncating and renormalizing
either sequence changes it by at most $2r_\Delta^{L_{\rm io}}$.  Hence
\begin{equation}\label{eq:truncation-error}
 \norm{
 \widetilde w_{\Delta,\tau,L_{\rm io}}\iota_{\rm grid}
 -\iota_{\rm grid}w_{\Delta,G}}
 \lesssim r_\Delta^{L_{\rm io}}.
\end{equation}
This concerns the encoded logical subspace; the implemented operator
remains the full unitary \cref{eq:finite-coupling} on the input and
output intervals and all completion modes.

\emph{Step 4: an isometric grid representation of reachable
histories.}
The scaled forms of \crefrange{eq:B}{eq:A0} show that an insertion
through $B$ produces $e^{-x}$ times a constant, while every subsequent
Volterra operation increases the polynomial degree by at most one.
Pointwise operations on $\EE$ change only the coefficients.  This proves
\cref{eq:scaled-reachable}.

For $q=e^{-x}f(x)\in\mathcal F^{\rm sc}_{R-1}(\tau)$, a weighted
polynomial
inverse estimate gives
\begin{equation}\label{eq:reachable-derivative}
 \norm{q'}_{L^2(0,\tau)}
 \lesssim L_R(\tau)\norm{q}_{L^2(0,\tau)}.
\end{equation}
Indeed, on $[-1,1]$ the normalized Legendre polynomial $\varphi_k$
satisfies
\[
 \norm{\varphi_k'}_2^2=\frac{2k+1}{2}k(k+1).
\]
Expanding a polynomial of degree below $R$, applying
Cauchy--Schwarz, and scaling to an interval of length $\ell$ gives
$\norm{f'}_2\lesssim R^2\ell^{-1}\norm{f}_2$.  On intervals of length
at most one the weight $e^{-2x}$ varies by an absolute factor, and
$q'=e^{-x}(f'-f)$.  Apply the estimate on all of $(0,\tau)$ when
$\tau<1$; when $\tau\ge1$, partition into intervals of length between
$1/2$ and
one and add the squared estimates.  Expansion in an orthonormal basis
gives the Hilbert-valued version.

Let $\mathsf I_\Delta:\bbc^G\to L^2(0,\tau)$ map $q_j$ to the constant
$q_j/\sqrt\Delta$ on cell $j$, let
$\operatorname{Av}_\Delta:=\mathsf I_\Delta^\dagger$, and put
$\Pi_\Delta:=\mathsf I_\Delta\operatorname{Av}_\Delta$.  Cellwise
Poincar\'e gives
\begin{equation}\label{eq:cell-projection}
 \norm{(\id-\Pi_\Delta)q}
 \lesssim\Delta L_R(\tau)\norm{q}.
\end{equation}
Restrict the averaging map to the reachable space,
\begin{equation}\label{eq:restricted-average}
 \operatorname{Av}_{\Delta,R}
 :=\operatorname{Av}_\Delta
 \big|_{\mathcal F^{\rm sc}_{R-1}(\tau)},
\end{equation}
where its adjoint is taken relative to that domain.  For every
$q\in\mathcal F^{\rm sc}_{R-1}(\tau)$,
\[
 \left\langle q,
 (\id-\operatorname{Av}_{\Delta,R}^\dagger
 \operatorname{Av}_{\Delta,R})q
 \right\rangle
 =\norm{(\id-\Pi_\Delta)q}^2.
\]
Consequently,
\begin{equation}\label{eq:A-near-isometry}
 \norm{\id-\operatorname{Av}_{\Delta,R}^\dagger
 \operatorname{Av}_{\Delta,R}}
 \lesssim\Delta^2 L_R(\tau)^2.
\end{equation}
Under the hypothesis of the lemma, the polar isometry
\begin{equation}\label{eq:Edelta}
 \mathsf E_\Delta
 =\operatorname{Av}_{\Delta,R}
 (\operatorname{Av}_{\Delta,R}^\dagger
  \operatorname{Av}_{\Delta,R})^{-1/2}
\end{equation}
is well defined and satisfies
\begin{equation}\label{eq:Edelta-close}
 \norm{\mathsf E_\Delta-\operatorname{Av}_{\Delta,R}}
 \lesssim\Delta^2 L_R(\tau)^2.
\end{equation}
It acts only on the history coordinate, proving the intertwining claim
in the lemma.  Notice that $\mathsf E_\Delta$ is a comparison isometry used in
the proof, not a circuit primitive.

\emph{Step 5: consistency of the continuous and cell evolutions.}
For a cellwise-constant field, define $\psi_j:=\psi(j\Delta)$ for
$0\le j\le G$.  Exact continuous evolution at cell boundaries and the
averaged output satisfy
\begin{align}
 \psi_{j+1}&=r_\Delta\psi_j-\iu\widehat\sigma_\Delta q_j,
 &\widehat\sigma_\Delta&=
 \frac{\sqrt2(1-r_\Delta)}{\sqrt\Delta},
 \label{eq:continuous-cell-state}\\
 (\operatorname{Av}_\Delta g)_j
 &=\nu_\Delta q_j-\iu\widehat\sigma_\Delta\psi_j,
 &\nu_\Delta&=\frac{2(1-r_\Delta)}{\Delta}-1.
 \label{eq:continuous-cell-output}
\end{align}
Uniformly for $0<\Delta\le1$, Taylor expansion gives
\begin{equation}\label{eq:cell-coeff-errors}
 |\sigma_\Delta-\widehat\sigma_\Delta|\lesssim\Delta^{5/2},
 \qquad
 |r_\Delta-\nu_\Delta|\lesssim\Delta^2.
\end{equation}
The stable geometric recurrence obeys
\begin{equation}\label{eq:stable-state-bound}
 \norm{(\psi_j)_{j=0}^{G-1}}_{\ell^2}
 \lesssim\Delta^{-1/2}\norm{(\xi,q)}.
\end{equation}
If $e_j=\psi_j-\zeta_j$, then
$e_{j+1}=r_\Delta e_j
-\iu(\widehat\sigma_\Delta-\sigma_\Delta)q_j$.  Convolution with the
geometric impulse response, using its $\ell^1$ norm for the whole
sequence and its $\ell^2$ norm for the final cell, gives
\begin{equation}
 \norm{(e_j)_{j=0}^{G-1}}_{\ell^2}
 \lesssim\Delta^{3/2}\norm{q},
 \qquad
 |e_G|\lesssim\Delta^2\norm{q}.
\end{equation}
Subtracting \cref{eq:cell-recurrence} from
\crefrange{eq:continuous-cell-state}{eq:continuous-cell-output}
leaves the terms
\[
 (\nu_\Delta-r_\Delta)q_j,\qquad
 -\iu(\widehat\sigma_\Delta-\sigma_\Delta)\psi_j,\qquad
 -\iu\sigma_\Delta e_j.
\]
\cref{eq:cell-coeff-errors,eq:stable-state-bound} yield
\begin{equation}\label{eq:cell-consistency}
 \norm{(\id\oplus\operatorname{Av}_\Delta)w_\tau
 (\id\oplus\mathsf I_\Delta)-w_{\Delta,G}}
 \lesssim\Delta^2.
\end{equation}

For an exact input in the public space plus
$\mathcal F^{\rm sc}_{R-2}(\tau)\otimes\EE$, insert
$\Pi_\Delta\otimes\id_\EE$ on the field and use
\cref{eq:cell-projection}.  Apply the $\EE$-valued extension of
\cref{eq:cell-consistency}: tensor the scalar comparison with
$\id_\HH$ on the active sector $V\HH$, while the continuous and finite
couplings are both the identity on $(V\HH)^\perp$, so the comparison
on the inactive sector is exact.  Then
replace cell averaging by its polar isometry on both input and output.
The output belongs to the public space plus
$\mathcal F^{\rm sc}_{R-1}(\tau)\otimes\EE$, so the same
$\mathsf E_\Delta^{(\EE)}$ applies.  Since the lemma assumes
$\Delta L_R(\tau)\le c_0$, these replacements give
the first two terms in \cref{eq:coupling-comparison}.  Finally,
\cref{eq:truncation-error} contributes
$O(r_\Delta^{L_{\rm io}})=O(e^{-\Delta L_{\rm io}})$.
\end{proof}

\subsection{Grid choice and the forward hybrid estimate}

Let
\begin{equation}\label{eq:eta-grid}
 \eta=\frac{c_{\rm grid}\eps}{R}
\end{equation}
for a sufficiently small absolute constant.  Let $G$ be the least even
integer satisfying $G\ge2$ and the following inequality, where
$C_{\rm grid}$ is a sufficiently large absolute constant:
\begin{equation}\label{eq:grid-choice}
 G\ge 2\tau+
 \frac{C_{\rm grid}\tau L_R(\tau)}{\eta},
 \qquad
 \Delta=\frac\tau G,
 \qquad
 L_{\min}=\left\lceil
 \Delta^{-1}\log\frac{C_{\rm grid}}\eta
 \right\rceil.
\end{equation}
Take the least power of two $M_{\rm cyc}\ge G+2L_{\min}$, choosing
$G$ even so that
$L_{\rm io}:=(M_{\rm cyc}-G)/2$ is an integer.  Then the right
side of \cref{eq:coupling-comparison} is $O(\eta)$.

Define the global direct-sum comparison isometry on reachable states by
\begin{equation}\label{eq:global-grid-embedding}
 \mathsf E_\Delta^{\rm glob}
 :=\id_{\PP_R}\oplus\mathsf E_\Delta^{(\EE)},
\end{equation}
followed by zero-amplitude inclusion into the completion modes of the
shared finite private space
\begin{equation}\label{eq:finite-private-space}
 \widetilde\QQ
 :=\Span\{\ket j:
 j\in\mathbb Z_{M_{\rm cyc}}\setminus\{\mathsf p\}\}\otimes\EE .
\end{equation}
Thus $\mathsf E_\Delta^{\rm glob}$ is the identity on every public slot and
acts as $\mathsf E_\Delta^{(\EE)}$ on the reachable private history,
which is supported on the $G$ interaction addresses.  The remaining
$2L_{\rm io}-1$ private addresses are the completion modes.  By
\cref{eq:grid-intertwining}, $\mathsf E_\Delta^{\rm glob}$ intertwines every
pointwise operator on $\EE$ with its finite-grid action.

For clarity, one binary address register implements the direct sum
$\PP_R\oplus\widetilde\QQ$.  Use the $M_{\rm cyc}$ cyclic addresses
and $R-1$ additional addresses, set
\[
 \mathsf p_0=\mathsf p=M_{\rm cyc}-1,
 \qquad
 \mathsf p_k=M_{\rm cyc}+k-1\quad(1\le k<R),
\]
and encode the public vector in slot $k$ as
$\ket{\mathsf p_k}Vx_k$, and the private vector at address $j$ as
$\ket j q_j$ for $j\ne\mathsf p$.  Public inactive address sectors
$\ket{\mathsf p_k}\otimes(V\HH)^\perp$ and any binary padding addresses
are unused identity sectors.  The $V$ encoding and decoding require
only a uniform-list-index preparation and its inverse.

At routed interaction $k$, reversibly swap the address labels
$\mathsf p_k$ and $\mathsf p_0$.  This permutation fixes every private
address and places the selected public slot at $\mathsf p$; it costs
polynomially many gates in $\log(M_{\rm cyc}+R)$ using comparisons and
controlled bit flips.  Apply the local cyclic coupling on the first
$M_{\rm cyc}$ addresses, with the active-list-index control described
above, act as the identity on all additional public addresses, and
undo the swap.  In particular, the phases $\mathsf D_\pm$ in
\cref{eq:encoding-completions} act only within this local cyclic
prefix and are the identity on the additional public addresses.
The Cayley query is applied only to private addresses,
including the private completion modes, and is the identity on
every public address.  Its pointwise action therefore preserves
$\PP_R\oplus\widetilde\QQ$.  The local coupling also preserves this
space, since its only public coordinate is
$\ket{\mathsf p}\otimes V\HH$ and the inactive channel is unchanged.
This gives a genuine finite unitary on one shared private register
for every selected public slot.  In each routed call, perform the
private Cayley query first, followed by this routed local coupling,
as in \cref{eq:S}.  Here and below, the exact interactions $S^{(k)}(1)$
and their product $\mathcal R_R(1)$ are understood after conjugation
by the isometric history rescaling in \cref{eq:scaled-coordinate}.  Let
$\widetilde S^{(k)}$ denote the finite version of $S^{(k)}(1)$ in
which the active scalar coupling $w_\tau$ is replaced by the $\EE$-valued
extension of $\widetilde w_{\Delta,\tau,L_{\rm io}}$ from
\cref{eq:finite-coupling}.  The routing and the Cayley query $U_M$
are unchanged, while the coupling remains the identity on
$(V\HH)^\perp$.  Define the
finite forward reuse circuit, and recall the exact circuit, by
\begin{equation}\label{eq:reuse-circuits}
 \mathcal R_R(1)
 =S^{(R-1)}(1)\cdots S^{(0)}(1),
 \qquad
 \widetilde{\mathcal R}_R
 :=\widetilde S^{(R-1)}\cdots\widetilde S^{(0)}.
\end{equation}
Thus $\widetilde w_{\Delta,\tau,L_{\rm io}}$ is one local coupling,
whereas $\widetilde{\mathcal R}_R$ is the complete $R$-interaction
routed circuit built from that coupling.  In particular,
$\widetilde{\mathcal R}_R$ is not the power
$(\widetilde w_{\Delta,\tau,L_{\rm io}})^R$: successive factors address
different public slots while reusing the same private grid.

The final circuit reuses the same private register $R$ times, so a
one-call estimate alone does not control the clean block: an early
discrepancy can propagate through every later interaction.  Each exact
forward history nevertheless remains in the polynomial class on which
the comparison is uniform, while the intervening operations are
unitary and obey the intertwining relation with the history embedding.
The next lemma
therefore upgrades the local coupling error to a linear-in-$R$
circuit error and carries it through clock projection and the
distinct input and output clock preparations.

\begin{lemma}[Forward reuse hybrid]\label{lem:forward-grid-hybrid}
On clean public-clock inputs,
\begin{equation}\label{eq:forward-grid-hybrid}
 \norm{\widetilde{\mathcal R}_R\mathsf E_\Delta^{\rm glob}
       -\mathsf E_\Delta^{\rm glob}\mathcal R_R(1)}
 \lesssim R\eta.
\end{equation}
Consequently, taking the input--output clock block preserves the
$O(R\eta)$ error bound.
\end{lemma}

\begin{proof}
At step $k$, evaluate \cref{lem:coupling-comparison} on the exact
history produced by the first $k$ calls.  This history lies in the
required polynomial space.  Previous discrepancies propagate through
a unitary, and \cref{eq:grid-intertwining} carries every intervening
Cayley query through $\mathsf E_\Delta^{\rm glob}$.
Induction and the triangle inequality therefore give
\cref{eq:forward-grid-hybrid}.  Clock preparation is an isometry and
clock projection is a contraction, so taking the input--output clock
block cannot increase this
operator-norm error.
\end{proof}

This is a forward comparison.  OAA uses the actual circuit
and its actual circuit adjoint.  Thus no approximation to
inverse continuous evolution on arbitrary, unreachable histories is
required.

The grid remains polynomially large.  From \cref{eq:grid-choice},
\begin{equation}\label{eq:G-polynomial}
 G=O\!\left(
 \tau+1+\frac{\max\{\tau,1\}R^2}{\eta}
 \right),
 \qquad
 \Delta^{-1}=G/\tau.
\end{equation}
For $\tau\ge1$, as ensured by the chosen working normalization,
\begin{equation}\label{eq:grid-log-size}
 \begin{aligned}
  1\le\tau,
  &\qquad
  R=O\!\left(
   \tau+L_0\log\frac{2}{\eps}
  \right),\\
  \eta^{-1}
  =O\!\left(
   \frac{\tau+L_0\log(2/\eps)}{\eps}
  \right),
  &\qquad
  \log M_{\rm cyc}+\log(1/\Delta)
  =O\!\left(
   \log\!\left[
    \frac{2}{\eps}
    \left(\tau+L_0\log\frac{2}{\eps}\right)
   \right]
  \right).
 \end{aligned}
\end{equation}

\subsection{Clock preparation and final accounting}

The normalized input and output amplitudes of the phase-twisted clock
pair are given explicitly in \cref{eq:clock-coefficients}.
Because the $c_j$ are real,
$b_j=\omega_{\rm tw}\overline{a_j}$.
Thus the two preparations use the same amplitude distribution and
differ only by known clock-index phases.

The integer coefficients of
\[
 \left(\sum_{n=0}^{K-1}z^n\right)
 \left(\sum_{n=0}^{m-1}z^n\right)^p
\]
can be computed by finite convolutions and have bit length
$O(\log K+p\log m)$, polynomial in
$\tau+L_0\log(2/\eps)$.  Consequently
$\lambda$, the attenuation rotation with clean matrix element
$\lambda/2$, and all state-preparation angles can be computed with
polynomial classical work.  A known length-$R$ complex amplitude
vector has a preparation
using $O(R)$ controlled rotations and phases \cite{Mottonen2005}.
Only constantly many full clock preparations occur in the three-call
OAA circuit; they lie outside the slot-by-slot reuse and
therefore cost
$\widetilde O(\tau+L_0\log(2/\eps))$, not
$\widetilde O((\tau+L_0\log(2/\eps))^2)$.
The public-address encoding $k\mapsto\mathsf p_k$ uses only comparisons
and reversible arithmetic on a binary word.  The unified address
register has $\lceil\log_2(M_{\rm cyc}+R-1)\rceil$ qubits, which is
$O(\log(2[\tau+L_0\log(2/\eps)]/\eps))$ by
\cref{eq:grid-log-size}.

\begin{proof}[Proof of \cref{thm:gate}]
For $t=0$, use the identity circuit.  Otherwise, run the logarithmic
query construction with analytic error at most $\eps/2$ and with the
working normalization from \cref{eq:working-normalization}, evaluated
at that accuracy.  Apply the preceding implementation with
$\Lambda$ replaced by $\Lambda'$ and $L_0$ replaced by $L_0'$.
Thus its evolution scale is $\tau'=t\sqrt d\,\Lambda'\ge1$ and
\cref{eq:normalization-tradeoff,eq:optimized-evolution-scale,eq:optimized-precision-scale}
give
\[
 R=O\!\left(\tau'+L_0'\log\frac{2}{\eps}\right)
 =O(Q_{\rm sim}).
\]
Replacing $\eps$ by a constant multiple changes $Q_{\rm sim}$ by at
most an absolute factor.  The working normalization may be rounded
upward within a constant factor, so its computation requires only
polynomial classical work in the supplied parameter bit lengths and
the working precision.

Each Cayley-query call uses $O(1)$ supplied oracle calls, canonical
ordering of the index pair, reversible arithmetic, and the rotation
\cref{eq:cayley-block}, whose scalar denominator is at least
$1/4$.  To apply the coupling, rotate the uniform list-index vector
$d^{-1/2}\sum_{\ell=1}^d\ket\ell$ to a distinguished basis flag,
apply \cref{eq:finite-coupling} only on that active channel, and undo
the rotation.  The explicit address swap above routes the
distinguished public slot.  The orthogonal list-index channel is left
invariant.  The input and output intervals, including their completion
modes, participate in the full unitary \cref{eq:finite-coupling}; any
amplitude entering them is covered by the comparison bound
\cref{eq:coupling-comparison}.

Allocate the grid error through \cref{eq:eta-grid} before OAA.
For the remaining implementation error, give each of the
$O(Q_{\rm sim})$ query-scale arithmetic or Fourier primitives operator
error $O(\eps/Q_{\rm sim})$, and give all clock preparations the same
total-order
budget.  A telescoping hybrid contributes $O(\eps)$.  By
\cref{eq:grid-log-size}, it suffices to use a working precision of
$P$ bits, with
\begin{equation}\label{eq:gate-precision}
 P=O\!\left(
 \log\frac{2Q_{\rm sim}}{\eps}
 +\log(n_H+b_{\rm val}+\log d+2)
 \right).
\end{equation}
The guard-bit workspace of the reversible arithmetic is
included in $a_{\rm ar}$.  The cost per query-scale step is
\begin{equation}
 O\!\left(g_F+g_H+
 \poly\!\left(
 n_H+b_{\rm val}+\log d+\log\frac{2Q_{\rm sim}}{\eps}
 \right)
 \right).
\end{equation}
Multiplying by $Q_{\rm sim}$, and adding the
$\widetilde O(Q_{\rm sim})$ clock preparations, proves
\cref{eq:gate-count}.  The system label, row and list-index registers, the
$b_{\rm val}$-qubit
matrix-value register, unified public--private address register,
constant amplification flags, and reusable workspaces give
\cref{eq:qubit-count}.  Every
output ancilla is included in the operator-norm error estimate.
\end{proof}

\begin{remark}
The conclusion is not that a sparse-oracle call or an arbitrary
prescribed rotation costs one elementary gate.  Rather, after charging
the supplied oracle circuits and reversible arithmetic explicitly,
the number of query-scale circuit blocks is linear in the logarithmic
query bound, and a polynomially large history grid appears only
through its logarithmic-size address.
\end{remark}

\clearpage
\bibliographystyle{alpha}
\bibliography{ref}

\clearpage
\appendix
\section{Proof of the one-round OAA lemma}
\label{app:purification}

\begin{proof}[Proof of \cref{lem:purification}]
Let $\mathcal A_0$ denote the supplied clean-block circuit.  Tensor it
with a known one-qubit rotation whose clean matrix element is
$\lambda/2$, and denote the resulting unitary by $\mathcal A$.  This
is possible because $\lambda\le2$.  Define the projector onto the
clean-ancilla subspace by
\[
 \Pi=\ket{0}_{\rm anc}\!\bra{0}_{\rm anc}\otimes\id.
\]
Under the natural identification of the range of $\Pi$ with the system
space,
\[
 X:=\Pi\mathcal A\Pi=\frac Q2.
\]

Define the clean-subspace reflection and the oppositely oriented
reflection about its image under $\mathcal A$ by
\[
 R_{\rm c}=2\Pi-\id,
 \qquad
 R_{\rm p}=\id-2\mathcal A\Pi\mathcal A^\dagger
          =-\mathcal A R_{\rm c}\mathcal A^\dagger.
\]
Both are reflections.  After the initial application of $\mathcal A$,
apply each reflection once:
\[
 \mathcal C
 :=R_{\rm p}R_{\rm c}\mathcal A
 =-\mathcal A R_{\rm c}\mathcal A^\dagger
     R_{\rm c}\mathcal A.
\]
The operators $\mathcal A$, $R_{\rm c}$, $R_{\rm p}$, and
$\mathcal C$ all act on the joint system and ancillary workspace.
Thus $\mathcal C$ is unitary.
Implementing $R_{\rm p}$ uses one call each to $\mathcal A^\dagger$
and $\mathcal A$.  Hence the construction uses three calls to the
given block circuit or its adjoint, together with known rotations and
clean-ancilla reflections.

Put $L=(\id-\Pi)\mathcal A\Pi$, so $\norm{L}\le1$.  Direct block
multiplication of the two reflections gives
\begin{align*}
 \Pi\mathcal C\Pi
 &=3X-4XX^\dagger X
   =\frac{3Q-QQ^\dagger Q}{2},\\
 (\id-\Pi)\mathcal C\Pi
 &=L(\id-4X^\dagger X)
   =L(\id-Q^\dagger Q).
\end{align*}
Let $D=\id-Q^\dagger Q$.  Since $U$ is unitary and
$E:=Q-U$ satisfies $\norm{E}\le\delta$,
\[
 D=-U^\dagger E-E^\dagger U-E^\dagger E.
\]
Consequently,
\[
 \norm{D}\le2\delta+\delta^2<\frac94\delta,
 \qquad
 \norm{Q}\le1+\delta<\frac54.
\]
Writing $F=(3Q-QQ^\dagger Q)/2=Q+QD/2$ for the clean output block,
we obtain
\[
 \norm{F-U}
 \le\delta+\frac12\norm{Q}\norm{D}
 <\frac52\delta,
 \qquad
 \norm{(\id-\Pi)\mathcal C\Pi}<\frac94\delta.
\]
The clean and non-clean output sectors are orthogonal.  Therefore
\[
 \norm{\mathcal C(\ket{0}_{\rm anc}\otimes\id)
       -\ket{0}_{\rm anc}\otimes U}
 \le
 \sqrt{\left(\frac52\delta\right)^2
       +\left(\frac94\delta\right)^2}
 =\frac{\sqrt{181}}4\delta
 <4\delta
\]
when $\delta>0$; for $\delta=0$ the left-hand side vanishes.  This
proves \cref{eq:purification}.
\end{proof}

\end{document}